\documentclass[a4paper,onecolumn,11pt,accepted=2023-07-31]{quantumarticle}
\pdfoutput=1
\usepackage[utf8]{inputenc}
\usepackage[english]{babel}
\usepackage[autostyle]{csquotes}
\usepackage[T1]{fontenc}
\usepackage{hyperref}
\usepackage{amssymb}
\usepackage{amsmath}
\usepackage{array}
\usepackage{geometry}
\usepackage{amsthm}
\usepackage{latexsym,bbm}
\usepackage[ruled,vlined,linesnumbered]{algorithm2e}
\usepackage{enumerate}
\usepackage{float}
\usepackage{multirow}
\usepackage{fullpage}
\usepackage{makecell}
\usepackage{tikz}

\newcommand{\bm}[1]{\mathbf{#1}}
\newcommand{\Ord}[1]{\mathcal{O}\left( #1 \right)}
\newcommand{\tOrd}[1]{\widetilde{\mathcal{O}}\left( #1 \right)}

\newcommand{\Hyp}[1]{\textit {\textbf {#1}}}

\newcommand{\bra}[1]{\left \langle{#1} \right |}
\newcommand{\ket}[1]{\left |{#1} \right \rangle}

\newcommand{\err}[0]{\mathrm{err}}

\newtheorem{lemma}{Lemma}
\newtheorem{theorem}{Theorem}
\newtheorem{corollary}{Corollary}
\newtheorem{fact}{Fact}
\newtheorem{oracle}{Data Input}

\newtheorem{defn}{Definition}

\def\be{\begin{eqnarray}}
\def\ee{\end{eqnarray}}

\usepackage{color}
\definecolor{Pr}{rgb}{0.4,0.3,0.9}

\newcommand{\E}{\mathbb{E}}

\newcommand\super[1]{^{#1}}
\usepackage[backend=bibtex]{biblatex}
\begin{document}

\title{Quantum Alphatron: quantum advantage for learning with kernels and noise}

\author{Siyi Yang}
\email{e0016915@u.nus.edu}
\affiliation{Centre for Quantum Technologies, National University of Singapore, Singapore 117543}
\author{Naixu Guo}
\email{naixug@u.nus.edu}
\orcid{0000-0001-6694-6454}
\affiliation{Centre for Quantum Technologies, National University of Singapore, Singapore 117543}
\author{Miklos Santha}
\affiliation{Centre for Quantum Technologies, National University of Singapore, Singapore 117543}
\affiliation{Universit\'e de Paris, IRIF, CNRS, F-75013 Paris, France; MajuLab UMI 3654}
\email{cqtms@nus.edu.sg}
\author{Patrick Rebentrost}
\affiliation{Centre for Quantum Technologies, National University of Singapore, Singapore 117543}
\email{cqtfpr@nus.edu.sg}
\maketitle

\begin{abstract}
At the interface of machine learning and quantum computing, an important question is what distributions can be learned provably with optimal sample complexities and with quantum-accelerated time complexities. In the classical case, Klivans and Goel discussed the \textit{Alphatron}, an algorithm to learn distributions related to kernelized regression, which they also applied to the learning of two-layer neural networks. 
In this work, we provide quantum versions of the Alphatron in the fault-tolerant setting.
In a well-defined learning model, this quantum algorithm is able to provide a polynomial speedup for a large range of parameters of the underlying concept class. We discuss two types of speedups, one for evaluating the kernel matrix and one for evaluating the gradient in the stochastic gradient descent procedure. We also discuss the quantum advantage in the context of learning of two-layer neural networks. Our work contributes to the study of quantum learning with kernels and from samples. 
\end{abstract}

\section{Introduction}
Machine learning is highly successful in a variety of applications using heuristic approaches even though the methods 
being used are often without strong guarantees on their learning performance. Important questions are why common machine learning algorithms such as stochastic gradient descent and kernel methods \cite{SS02} work well and what is the best way to interpret the results. 
Computational learning theory addresses some of the fundamental theoretical questions and provides a systematic framework to discuss provable learning of probability distributions and machine learning architectures (such as neural networks). 
In a variety of settings and architectures, further assumptions on the underlying distribution can rule out hard instances and lead to provable and fast learning algorithms. Such guarantees have been given for generalized linear models, Ising models, and Markov Random Fields \cite{Klivans2017}, for example.
The \textsc{Alphatron} developed by Goel and Klivans \cite{pmlr-v99-goel19b} is a gradient-descent like algorithm for isotonic regression with the inclusion of kernel functions, which provably learns a kernelized, non-linear concept class of functions with a bounded noise term.
As a consequence it can be employed to learn two-layer neural networks, where one layer of activation functions feeds into a single activation function. 

Quantum machine learning has received a great deal of attention \cite{biamonte2016quantum, CHI+18j}, with the hope of gaining
advantages for problems such as regression \cite{harrow2009quantum, PhysRevLett.109.050505, PhysRevLett.113.130503} and neural networks (e.g. \cite{10.1145/3411466, cao2017quantum}). 
Quantum gradient computation has been considered in \cite{Gilyen2019gradient}.
Many algorithms are envisioned for near-term 
quantum computers \cite{Preskill2018quantumcomputingin, mcclean16,Beer2020}.
Kernel methods extend linear learning tasks to non-linear learning tasks and, similarly, quantum kernel methods \cite{Schuld2019,havlicek2019} use a high-dimensional Hilbert space of a quantum system to encode features of data.
Some algorithms are similar in spirit to the use of heuristic methods in classical machine leaning. They often cannot obtain provable guarantees for the quality of learning and for the run time.  
An interesting avenue for quantum algorithms for machine learning is therefore to take provable classical algorithms for learning and study provable quantum speedups which retain the guarantees of the classical algorithms \cite{BS17c,vAG19p,LCW19}. 

In this work, we provide quantum speedups for the \textsc{Alphatron} and its application to non-linear classes of functions and two-layer neural networks.
First, we consider the simple idea of pre-computing the kernel matrix used in the algorithm.  Our setting is one where the samples are given via quantum query access. Using this access, we can harness quantum subroutines to estimate the entries of the kernel matrices used in the algorithm. 
The quantum subroutines we use are adaptations of the amplitude estimation algorithm. 
We show that the learning guarantees can be preserved despite the erroneous estimation of the kernel matrices. 
In a subsequent step, we also quantize the \textsc{Alphatron} algorithm itself. To this end, we require the storage of intermediate values in a quantum-accessible memory \cite{Giovannetti2008,Giovannetti2008_2, Arunachalam2015}. In particular, we show that there are estimations inside the main loop of the algorithm which can be replaced with quantum subroutines, while keeping the main loop intact. We carefully study the regimes where the algorithms allow for a quantum speedup.
We are again able to show that the others parameters of the algorithm remain stable under these estimations. 
Our main result is that we obtain a quantum algorithm for  learning the original concept class, where the a quantum speedup is obtained for a large parameter regime of the concept class. 
We consider a previously defined class of two-layer neural networks and demonstrate that these networks are outside the regime for quantum advantage. We define a different neural network architecture which can exhibit a quantum advantage for learning. 

The paper is organized as follows.
Section \ref{sec:preliminary} discusses the mathematical preliminaries, the weak p-concept learning setting, and kernel methods, and introduces the Alphatron algorithm with a run time analysis. 
Section \ref{sec:four} discusses the kernel matrix estimation in the context of the Alphatron using both classical sampling and quantum estimation.   
Section \ref{sec:five} discusses the main loop of the Alphatron and the corresponding quantum run time. 
Section \ref{sec:discuss} summarizes the results in terms of all relevant parameters and discusses the regime where a quantum speedup is obtained.
Finally Section \ref{sec:application} considers the application for learning two-layer neural networks, where we give a neural network architecture that may achieve a quantum speedup.

\section{Preliminaries and Alphatron algorithm}
\label{sec:preliminary}
\subsection{Notations}
The vectors are denoted by bold-face $\mathbf {x}$, and their elements by $x_j$. We leave in plain font the $\alpha$ vector (and all other vectors denoted with Greek symbols).
The standard vector space of reals and the unit ball of dimension $n$ are denoted by $\mathbb{R}^n$ and $\mathbb{B}^{n}$, respectively.
The $\ell_p$-norm of vectors in $\mathbb{R}^n$ is denoted by $\Vert \cdot \Vert_p$.
Moreover, the max norm is denoted by $\Vert \bm{x} \Vert_{\max} = \max_i |\bm{x}_i|$.
We use $\mathbf{a} \cdot \mathbf{b}$ to denote the standard inner product in $\mathbb{R}^n$.
We use the notation $\tOrd{}$ to omit any poly-log factors in the arguments.
When we write $g +\Ord{\dots}$, we mean $g+f$ with some $f\in \Ord{\dots}$.
We use  $a:=b$ to define $a$ in terms of $b$.
Given two random variables $X$ and $Y$, denote by $\mathbb E[Y|X]$ the conditional expectation value.
\subsection{Cost of arithmetic operations}
We use the following arithmetic model for classical computation. We represent the real numbers with a sufficiently large number of bits. We assume that the number of bits is large enough
to make the numerical errors negligible
in the correctness and run time proofs of the algorithms under consideration.
The implication is that we can ignore numerical errors
of arithmetic operations (e.g., addition, subtraction, multiplication, and so on) with respect to truncation or rounding.
Hence, we assume all real numbers cost $\Ord{1}$ space and
the basic arithmetic operations between them cost $\Ord{1}$ time.
While the accumulated error can be important, dealing with a proper error analysis would require
a substantial deviation from the main purpose of this paper.

For the quantum algorithms, we keep track of the amount of (quantum) bits for storing real numbers. 
We use a standard fixed-point encoding of  real numbers. 
\begin{defn}[Notation for encoding of real numbers] \label{defEncoding}
Let $c_1,c_2$ be positive integers and $a\in \{0,1\}^{c_1}$ and $b\in \{0,1\}^{c_2}$ be bit strings. 
Define the (signed) rational number 
\be
\mathcal Q(a,b,s) :=
(-1)^s \left( 2^{c_1-1}a_{c_1} + \cdots + 2 a_2 + a_1+ \frac{1}{2} b_1 + \cdots + \frac{1}{2^{c_2}} b_{c_2}\right) \in [-R,R],
\ee
where $R:=2^{c_1}-\frac{1}{2^{c_2}}$.
For representing numbers in $[0,2 -\frac{1}{2^{c}}]$ with positive integer $c$ bits after the decimal point (the case $[0,1]$ being the most frequently used in this work)  we use $c_1=1$ and $c_2 = c$, and define the short-hand notation
\be
\mathcal Q(z) := \mathcal Q(a,b,0)= a + \frac{1}{2} b_1 + \cdots + \frac{1}{2^{c}} b_{c},
\ee
where $z := (a,b) \in \{0,1\}^{c+1}$.
Given a vector of bit strings $\bm v\in (\{0,1\}^{c+1})^n$, the notation $\mathcal Q (\bm v)$ means the vector whose $j$-th component is
$\mathcal Q (v_j)$.
\end{defn}
For any real number $r \in [0,2^{c_1}]$ there exist $a\in \{0,1\}^{c_1}$ and $b\in \{0,1\}^{c_2}$  such that the difference to $\mathcal Q(a,b,0)$ is at most $\frac{1}{2^{c_2+1}}$.

\subsection{The learning model}
We consider the standard ``probabilistic concept" (p-concept) learning model (Ref.~\cite{Kearns1994}) in our paper.
Let $\mathcal{X}$ be the input space and $\mathcal{Y}$ be the output space. A \textit{concept class}
$\mathcal{C}$ is a class of functions mapping the input space to the output space, i.e.,  $\mathcal{C} \subseteq \mathcal{Y}^{\mathcal{X}}$. 
We define here weak learnability with a fixed lower bound for the error, in contrast to the standard definition of p-concept learnability for all $\epsilon_0 > 0$.
\begin{defn}[Weak p-concept learnable] \label{defLearnable}
For $\epsilon_0> 0$, a concept class
$\mathcal{C}$ is ``weak \textit{p-concept learnable} up to $\epsilon_0$'' if 
there exists
an algorithm $\mathcal{A}$ such that for every $\delta>0$, $c \in \mathcal{C}$, and distribution
$\mathcal{D}$ over $\mathcal{X} \times \mathcal{Y}$ with
$\mathbb{E}_{y} [y | \mathbf{x}] = c (\mathbf{x})$ we have
that $\mathcal{A}$, given access to samples drawn from $\mathcal{D}$, outputs a \textit{hypothesis} $h:\mathcal{X} \to \mathcal{Y}$,
such that with probability at least $1-\delta$,
\[
\varepsilon(h) := \mathbb{E}_{(\mathbf{x}, y) \sim \mathcal{D}} \left[ (h(\mathbf{x}) - c (\mathbf{x}))^2 \right] \le \epsilon_{0}.
\]
\end{defn}
The quantity $\varepsilon(h)$ is the generalization error of hypothesis $h$.
Moreover, if we have $m$ samples $(\bm{x}_i, y_i)$ drawn from the distribution $\mathcal{D}$, 
we define the empirical error of $h$ as
\[
\hat{\varepsilon}(h) := \frac{1}{m} \sum_{i=1}^m (h(\mathbf{x}_i) - c (\mathbf{x}_i))^2.
\]
For convenience, we also define another similar function as
\[
\err(h) := \mathbb{E}_{(\bm{x}, y) \sim \mathcal{D}} [(h(\bm{x}) - y)^2].
\]
Since $\mathbb{E}_{(\bm{x}, y) \sim \mathcal{D}} [y] = \mathbb{E}_{(\bm{x}, y) \sim \mathcal{D}} [\mathbb{E}_y [y | \bm{x}]]$,
it is easy to see that
\be
\begin{array}{rcl}
\err(h) - \err(\mathbb{E}_y [y | \bm{x}]) & = & \mathbb{E}_{(\bm{x}, y) \sim \mathcal{D}} [(h(\bm{x}) - y)^2] -
\mathbb{E}_{(\bm{x}, y) \sim \mathcal{D}} [(\mathbb{E}_y [y | \bm{x}] - y)^2]\\
& = & \mathbb{E}_{(\bm{x}, y) \sim \mathcal{D}} [h(\bm{x})^2 - 2\mathbb{E}_y [y | \bm{x}]h(\bm{x}) + \mathbb{E}_y [y | \bm{x}]^2 ]\\
& = & \mathbb{E}_{(\bm{x}, y) \sim \mathcal{D}} [h(\bm{x})^2 - 2c(\bm{x})h(\bm{x}) + c(\bm{x})^2]
=  \varepsilon(h).
\end{array}
\label{eq:constErr}
\ee
Note that $\mathbb{E}_{y} [y | \bm{x}]$ is independent of the choice of $h$.
Hence, for hypotheses $h_1$ and $h_2$, 
we have
$\err(h_1) - \err(h_2) = \varepsilon(h_1) - \varepsilon(h_2)$.
Thus, we may use $\err()$ instead of $\varepsilon()$ for comparing hypothesis.
Moreover, by using the empirical version of the $\err(h)$ function, even without knowing the probability distribution $\mathcal{D}$,
we are still able to evaluate the quality of the hypothesis $h$ given $m$ samples $(\mathbf{x}_i, y_i) \sim \mathcal{D}$ as
\[
\hat{\err}(h) := \frac{1}{m} \sum_{i=1}^m (h(\mathbf{x}_i) - y_i)^2.
\]
By the Chernoff bound,
we may bound the generalization error $\err()$ in terms of the empirical error $\hat{\err}()$ with high probability.

To learn a good hypothesis, on the one hand, we prefer to assume a relatively simple concept class (e.g.,
a concept class consisting only of linear functions). Then it is easy to design 
an algorithm for finding the best hypothesis in
that class.
On the other hand the real-world data distribution is often complicated and cannot be captured by 
a hypothesis from a simple concept class.
The kernel trick is widely used to turn a 
simple linear concept class and a given learning algorithm into a non-linear concept class and a corresponding learning algorithm, usually without
changing too much the algorithm. 
In the kernel method, we use a more general function to measure the similarity between two vectors instead of the linear inner product.
The kernel function $\mathcal{K}: \mathcal{X} \times \mathcal{X} \to \mathbb{R}$ is a
(usually non-linear) similarity measure
on the input space and it is defined via a feature map.
Let $\mathcal{V}$ be an arbitrary metric space with inner product $\langle \cdot, \cdot \rangle$.
The feature map $\psi : \mathcal X \to \mathcal{V}$
maps any input vector into the metric space (also called feature space).
For vectors $\mathbf{x}, \mathbf{y} \in \mathcal{X}$,
we define $\mathcal{K}(\mathbf{x}, \mathbf{y}) = \langle \psi(\mathbf{x}), \psi(\mathbf{y}) \rangle$.

For our purpose, we use the multinomial kernel function to allow the learning of non-linear concepts. 
Consider formal polynomials over $n$ variables of total degree $d$.
There are $n_d :=\sum_{i=0}^d n^i$ monomials, which can be uniquely indexed by the tuples $(k_1, \cdots, k_i) \in [n]^i$ for $i\in \{0, \cdots, d\}$.
We consider from now on the input space $\mathcal{X} = \mathbb{B}^{n} \subseteq \mathbb{R}^n$ and the feature space $\mathcal V = \mathbb R^{n^d}$.
We consider the standard Euclidean metric on the feature space $\mathbb{R}^{n_d}$.
We define a normalized feature map $\psi_d : \mathbb{B}^{n} \to \mathbb{R}^{n_d}$ which maps a vector
$\mathbf{x} = (x_1, x_2, \cdots, x_n)^T \in \mathcal X$
to a $n_d$-dimensional vector of monomials of the variables $x_1, x_2, \cdots, x_n$. For $i\in \{0, \cdots, d\}$ and $(k_1, \cdots, k_i) \in [n]^i$, define
\be
[\psi_d (\bm x)]_{(k_1, \cdots, k_i)} := \frac{1}{\sqrt {d+1}} \prod_{j=1}^{i} \mathbf{x}_{k_j}.
\ee
When $i = 0$, we have the empty tuple and the corresponding component is the constant term $1$. Note that for $n\geq2$, we have both the components $x_1x_2$ and $x_2 x_1$. This redundancy can be avoided by the use of ordered multisets but will not influence our discussion. 
For $\bm{x}, \bm{y} \in \mathbb{B}^{n}$,
we can compute the inner product as
\[ \langle \psi_d(\mathbf{x}) , \psi_d(\mathbf{y}) \rangle = \frac{1}{d+1} \sum_{i=0}^d \sum_{1 \le k_1,\cdots,k_i \le n} \left [\left (\prod_{j=1}^i \mathbf{x}_{k_j}\right) \cdot \left(\prod_{j=1}^i \mathbf{y}_{k_j}\right )\right]
= \frac{1}{d+1} \sum_{i =0}^d (\mathbf{x} \cdot \mathbf{y})^i.\]
Observe that by definition, for all $\textbf{x}, \textbf{y} \in \mathbb{B}^{n}$, we have that  
$\langle \psi_d(\mathbf{x}), \psi_d(\mathbf{y})\rangle \leq 1$. 
This can be seen from,
\[ \max_{\textbf{x} \in \mathbb{B}^{n}}
\langle \psi_d(\mathbf{x}), \psi_d(\mathbf{x})\rangle = \frac{1}{d+1} \max_{\textbf{x} \in \mathbb{B}^{n}} \sum_{i=0}^d
\left(\Vert\textbf{x}\Vert_2^2 \right)^i =1, 
\]
where we have used that $\max_{\textbf{x} \in \mathbb{B}^{n}} \Vert\textbf{x}\Vert_2^2 = 1$ for the unit ball $\mathbb{B}^{n}$. 
Throughout this paper, our definition for the normalized multinomial kernel function with degree $d$ is
\[ \mathcal{K}_d(\mathbf{x}, \mathbf{y}) := \langle \psi_d(\mathbf{x}), \psi_d(\mathbf{y}) \rangle =
\frac{1}{d+1}\sum_{i =0}^d (\mathbf{x}\cdot \mathbf{y})^i.\]
With these definitions, let us consider the following concept class.
\begin{defn} [Concept class and distribution] \label{assume:alpha}
Let $\mathcal{K}_d$ be the normalized multinomial kernel function corresponding to the feature map $\psi_d:\mathbb{B}^{n} \to \mathbb{R}^{n_d}$, defined as above. 
Let $L, B,\zeta,\epsilon >0$. 
Let $u: \mathbb{R} \rightarrow [0,1]$ be a known $L$-Lipschitz non-decreasing function.
Define the concept class $\mathcal C \subseteq [0,1]^{\mathbb{B}^{n}}$ of functions, where for $c \in \mathcal C$ there exists a $\textbf{v} \in \mathbb R^{n_d}$ with $\Vert \textbf{v}\Vert_2 \leq B$ and a $\xi: \mathbb{R}^n \rightarrow [-\zeta, \zeta]$ (noise function) with $\mathbb{E}_{\mathbb{B}^{n}}\left[\xi(\bm{x})^2 \right] \leq \epsilon$ such that
\be
c(\bm x) = u(\langle \textbf{v}, \psi_d(\bm{x}) \rangle + \xi(\bm{x})).
\ee
Consider a distribution $\mathcal{D}$ on $\mathbb{B}^{n} \times [0,1]$
for which $c\in \mathcal C$ exists such that the distribution satisfies 
$$\mathbb{E}_y[y | \bm{x}] = c(\bm x),$$
where $\mathbb{E}_y[y | \bm{x}]$ is the conditional expectation value.
\end{defn}

The $\epsilon$ in this definition motivates  Definition \ref{defLearnable}, as we will see that we cannot learn the concept class for all $\epsilon_0>0$. The intrinsic error $\epsilon$ will define a lower bound for the error.
The learning guarantee for the \textsc{Alphatron} and all algorithms in this work is proven for this concept class. 

\subsection{Alphatron algorithm}

We review the classical \textsc{Alphatron} algorithm of \cite{pmlr-v99-goel19b}.
Alphatron is the first provably efficient algorithm for learning neural networks with two nonlinear layers without further assumptions on the structure of the neural network.
It can also be used for many problems like Boolean Learning and Multiple Instance Learning.
The key idea is to combine  isotonic regression with kernel methods, where  isotonic regression is to find a non-decreasing function for predicting sequences of observations.
More explicitly, isotonic regression is to find
$$
\hat{\beta}=\arg\min_{\beta \in \mathbb R^m} \sum_{i=1}^m (x_i-\beta_i)^2\ \mathrm{subject\ to }\ \beta_1\leq\dots\leq \beta_m,
$$
where $x_i \in \mathbb R$ are a sequence of $m$ data.
The setting considered for the \textsc{Alphatron} algorithm includes a kernel function into the problem, as we can see in the following.

The setting is as follows.
We split the data into two parts, training set and validation set.
The training data set contains $(\mathbf{x}_i, y_i)^{m_1}_{i=1} \in \mathbb{B}^{n} \times [0,1]$, where $m_1$ is the size of the training data.
The validation data set contains $(\mathbf{a}_i, b_i)^{m_2}_{i=1}\in \mathbb{B}^{n} \times [0,1]$ of size $m_2$.
Let $m := m_1 + m_2$ be the total size of the data set.
Then, since $m_1, m_2 \in \Ord{m}$ we can use $\Ord{m}$ as an upper bound of the size of data.
In the \textsc{Alphatron} algorithm, we first build several hypotheses from the training set.
Then we use the validation set to evaluate each hypothesis and select the optimal one from them.
\begin{algorithm}[H]
\label{algAlphatron}
\caption{\textsc{Alphatron}}
\SetKw{KwReturn}{Return}
\SetKw{KwInput}{Input}
\SetKw{KwOutput}{Output}
\KwInput{training data $(\mathbf{x}_i, y_i)^{m_1}_{i=1}$, testing data $(\mathbf{a}_i, b_i)^{m_2}_{i=1}$, 
function $u: \mathbb{R} \to [0, 1]$, number of iterations $T$, degree of the multinomial kernel $d$, learning rate $\lambda$.}
\DontPrintSemicolon

$\alpha^0 \leftarrow \mathbf{0} \in \mathbb{R}^{m_1}$\;
\For {$t \leftarrow 0$ \KwTo $T - 1$} {
define $h^t(\bm{x}) := u\left(\sum_{i=1}^{m_1} \alpha_i^t \mathcal{K}_d(\mathbf{x}, \mathbf{x}_i)\right)$ \label{line:defineh}\;
\For {$i \leftarrow 1$ \KwTo $m_1$} {
$\alpha_i^{t+1} \leftarrow \alpha_i^t + \frac{\lambda}{m_1} (y_i - h^t(\mathbf{x}_i))$ \label{line:updatealpha}\;
}
}
\KwOutput{ $\bm \alpha^{t_{\rm out}}$, where $t_{\rm out} = \arg \min_{t \in [T]} \frac{1}{m_2} \sum_{j=1}^{m_2} (h^t (\bm{a}_j) - b_j)^2$} \label{line:valid}\;
\end{algorithm}

The algorithm has a number of iterations $T$, which will be related to the other input quantities via the learning guarantees.
In each iteration, the algorithm generates a new vector $\alpha^{t+1}$ and a new hypothesis $h^{t+1}$ from the old ones.
The output of the algorithm is a vector $\alpha ^{t_{\rm out}} \in \mathbb R^{m_1}$ describing the hypothesis $h^{t_{\rm out}}:\mathbb{B}^{n} \to [0,1]$,
which has some p-concept error according to the input data.
First, we are interested in the general run time complexity, before we discuss the guarantees for the weak p-concept learning of the specific concept class. 
\begin{theorem} [Run time of \textsc{Alphatron}] \label{thmOrig}
Algorithm~\ref{algAlphatron} has a run time of $\mathcal{O}(Tm^2(n+\log d))$.
\end{theorem}

\begin{proof}
For computing the multinomial kernel function $\mathcal{K}_d(\bm{x}, \bm{y})$,
we need to first compute the inner product $\langle \bm{x}, \bm{y} \rangle$ in $\mathcal{O}(n)$ time trivially.
For all $r\in \mathbb R \setminus \{1\}$, since $\sum_{0 \le i \le d} r^i = \frac{r^{d+1}-1}{r-1}$, it costs $\Ord{\log d}$ time to
compute the multinomial kernel function from the inner product $r$. Thus, by the definition of
$h^t(\bm{x})$ in line~\ref{line:defineh} of the algorithm, $h^t(\bm{x})$ is computed in $\mathcal{O}(m(n + \log d))$ for a given $\bm x$. In the first part, the training phase, 
line~\ref{line:updatealpha} is executed for $\mathcal{O}(Tm)$ times. Hence, this part costs $\mathcal{O}(Tm^2 (n+\log d))$.
Similarly, in the second part, the validation phase (line~\ref{line:valid}), a number $\mathcal{O}(T m)$ of calls to the function $h^t (\bm{a})$ is used.
Hence, the algorithm costs $\mathcal{O}(Tm^2(n+\log d))$ in total.
\end{proof}

For obtaining a learning guarantee in the setting from Definition \ref{assume:alpha}, it is  supposed that 
the following relations between the  parameters hold.
\begin{defn} [Parameter definitions and relations] \label{assume:param}
Consider the setting in Definition \ref{assume:alpha}, which defines the distribution $\mathcal D$ and the parameters $(B, L, \zeta, \epsilon)$, and consider the Algorithm \ref{algAlphatron}, which uses the parameters $(m_1, m_2, T, \lambda)$. 
Define the following additional parameters and fix the following relationships between the parameters. 
\begin{enumerate}
\item Equate the $L$-Lipschitz non-decreasing function from the concept class with the function $u$ used in the algorithm.
\item Learning rate $\lambda =1/L$.
\item Let the training set $(\bm{x}_i,y_i)_{i=1}^{m_1}$ be sampled iid from $\mathcal D$.
\item Let $C>0$ be a large enough constant and set 
$T = C B L\sqrt{m_1/\log(1/\delta)}$.
\item Let $C^\prime>0$ be a large enough constant and set 
$m_2 = C^\prime m_1\log(T/\delta)$, and let the validation set  $(\bm{a}_i,b_i)_{i=1}^{m_2}$ be sampled iid from $\mathcal D$.
\item Define $A_2 := L\sqrt{\epsilon} + L \zeta \sqrt[4]{\frac{\log(1/\delta)}{m_1}} + BL\sqrt{\frac{\log(1/\delta)}{m_1}}$ and let $C^{\prime\prime}>0$ be a large enough constant. 
\end{enumerate}
\end{defn}
The following learning guarantee was proven in Ref.~\cite{pmlr-v99-goel19b}.
\begin{theorem} [Learning guarantee of \textsc{Alphatron}, same as Ref.~\cite{pmlr-v99-goel19b}] \label{thmOrigLearning}
Given the learning setting in Definition \ref{assume:alpha}
and the parameters defined in Definition \ref{assume:param},
Algorithm~\ref{algAlphatron} outputs $\alpha^{t_{\rm out}}$ which describes the hypothesis $h^{t_{\rm out}}(\bm{x}) := u\left(\sum_{i=1}^{m_1} \alpha_i^{t_{\rm out}} \mathcal{K}_d(\mathbf{x}, \mathbf{x}_i)\right)$ such that with probability $1 - \delta$,
\[
\varepsilon\left(h^{t_{\rm out}}\right) \leq C^{\prime\prime} A_2.
\]
\end{theorem}
Next, we discuss a regime where Theorem \ref{thmOrigLearning} achieves p-concept learnability. This result was implicit in Ref.~\cite{pmlr-v99-goel19b}.
\begin{theorem}[P-concept learnability via the \textsc{Alphatron}] \label{thmLearning}
Let $m'_1 := \frac{16 \zeta^4}{ \epsilon^2} \log(1/\delta)$ and $m''_1 := \frac{ 4 B^2}{ \epsilon} \log(1/\delta)$.
If $m_1 \geq \max \left \{m'_1, m''_1 \right \},$  then the concept class in Definition \ref{assume:alpha} is weak p-concept learnable up to $2 C^{\prime\prime} L \sqrt \epsilon$ by the \textsc{Alphatron} algorithm.
\end{theorem}
\begin{proof}
Recall that weak p-concept learnability up to $\epsilon_0$ means that
\be
\epsilon(h^{t_{\rm out}}) \leq \epsilon_0,
\ee
with probability $1-\delta$. Hence, we desire that $C^{\prime\prime} A_2 \leq \epsilon_0$.
Note the trivial case when 
$\epsilon_0 < C^{\prime\prime} L \sqrt \epsilon$, which means that the intrinsic error of the concept class is too large, and we fail to achieve learnability. 
Hence, we can only prove the case when $\epsilon_0 \geq C^{\prime\prime} L \sqrt \epsilon$, 
and we prove the theorem only for $\epsilon_0 = 2 C^{\prime\prime} L \sqrt \epsilon$,
where we use a factor $2$ to leave room for the other terms in $A_2$.
It follows that we would like to set $m_1$ such that 
\be \label{eqtwotermsfinal}
C^{\prime\prime} \left( L \zeta \sqrt[4]{\frac{\log(1/\delta)}{m_1}} + BL\sqrt{\frac{\log(1/\delta)}{m_1}} \right)\leq \frac{1}{2}\epsilon_0.
\ee
Here, Eq.~(\ref{eqtwotermsfinal}) can be achieved by making each term smaller than $\epsilon_0/4$ (alternatively, we can solve a quadratic equation, which leads to more complicated equations).
This means that both of following statements have to be true:
\be
m_{1} &\geq& \frac{256 C^{\prime\prime,4} L^4 \zeta^4}{\epsilon_0^4} \log(1/\delta),\\
m_{1} &\geq& \frac{ 16 C^{\prime\prime,2} B^2 L^2}{\epsilon_0^2} \log(1/\delta).
\ee
Hence, we take $m_1$ greater than the maximum of the right-hand side expressions. 
Employing the lower-bound $2 C^{\prime\prime} L \sqrt \epsilon \leq \epsilon_0$, we find that 
$ m_1 \geq \max \left \{ m'_1,m''_1 \right \}$ 
leads to weak p-concept learnability up to $2C^{\prime\prime} L \sqrt \epsilon$. 
\end{proof}

\section{Pre-computation and approximation of the kernel matrix}\label{sec:four}
One bottleneck of the \textsc{Alphatron} algorithm is the repeated inner product computation when evaluating the function $h^t (\bm{x})$. 
In Algorithm~\ref{algAlphatron}, at every step $t$ out of $T$ steps, we need to evaluate $\Ord{m^2}$ inner products for the kernel function. This evaluation is redundant because the inner products do not change for different $t$.
A simple pre-computing idea helps to reduce
the time complexity to some extent. We improve Algorithm~\ref{algAlphatron} as follows.
Given input data $(\mathbf{x}_i, y_i)^{m_1}_{i=1}$ and $(\mathbf{a}_i, b_i)^{m_2}_{i=1}$ and $d$,
we define two matrices
\[ K_{ij} := \mathcal{K}_d(\mathbf{x}_i, \mathbf{x}_j),\]
and
\[ K'_{ij} := \mathcal{K}_d(\mathbf{a}_i, \mathbf{x}_j).\]
If these two matrices are given by an 
oracle, then we are able to rewrite Algorithm~\ref{algAlphatron} as below.
Algorithm~\ref{algAlphatronMain} will be used as a subroutine several times in the remainder of this paper.

\begin{algorithm}[H]
\label{algAlphatronMain}
\caption{\textsc{Alphatron\_with\_Kernel}}
\SetKw{KwReturn}{Return}
\SetKw{KwInput}{Input}
\SetKw{KwOutput}{Output}
\KwInput{
Function $u: \mathbb{R} \to [0, 1]$, number of iterations $T$, 
learning rate $\lambda$,
query access to $K_{ij}$ and $K'_{ij}$}
\DontPrintSemicolon

$\alpha^0 \leftarrow \mathbf{0} \in \mathbb{R}^{m_1}$\;
\For {$t \leftarrow 0$ \KwTo $T-1$} {
\For {$i \leftarrow 1$ \KwTo $m_1$} {
$\alpha_i^{t+1} \leftarrow \alpha_i^t + \frac{\lambda}{m_1} \left(y_i - u\left (\sum_{j=1}^{m_1} \alpha_i^t \cdot K_{ij}\right)\right)$\;
}
}
\KwOutput{$\alpha^{t_{\rm out}}$, where $t_{\rm out} = \arg \min_{t \in [T]} \frac{1}{m_2} \sum_{i=1}^{m_2} \left(u\left(\sum_{j=1}^{m_1} \alpha_i^t \cdot K'_{ij}\right) - b_i\right)^2$} \label{algLine:valid}\;
\end{algorithm}
With equivalent input, Algorithm~\ref{algAlphatronMain} produces the same output as Algorithm \ref{algAlphatron}, which can be easily checked as follows. 
Fix the input for Algorithm~\ref{algAlphatron}. From these fixed training examples compute the kernel matrices $K_{ij} = \mathcal{K}_d(\mathbf{x}_i, \mathbf{x}_j)$
and
$K'_{ij} = \mathcal{K}_d(\mathbf{a}_i, \mathbf{x}_j)$.
Use these matrices and the other inputs of Algorithm \ref{algAlphatron} to fix the input of Algorithm~\ref{algAlphatronMain}.
The sequences $(\alpha_{\rm Alg\ref{algAlphatron}}^t)_{t\in[T]}$ and 
$(\alpha_{\rm Alg\ref{algAlphatronMain}}^t)_{t\in[T]}$ of both algorithms are the same and hence for the output it holds that
$$ \alpha_{\rm Alg\ref{algAlphatron}}^{t_{\rm out,Alg1}} = \alpha_{\rm Alg\ref{algAlphatronMain}}^{t_{\rm out,Alg2}}.$$
Note that even if we do not explicitly define the hypothesis $h^t$ in Algorithm~\ref{algAlphatronMain},
in the analysis, we still use the same notation $h^t$ for the $t$-th generated hypothesis as in Algorithm~\ref{algAlphatron}.
\begin{theorem} [\textsc{Alphatron\_with\_Kernel}]\label{thmVanilla}
Algorithm~\ref{algAlphatronMain} runs in time $\Ord{Tm^2}$.
\end{theorem}
\begin{proof}
Since each entry of the matrices $K$ and $K'$ is accessible in $\Ord{1}$ time,
the run time of the algorithm is 
$\Ord{Tm^2}$.
\end{proof}

We now discuss the  pre-computation, i.e., we prepare 
the matrices
$K_{ij}$ and $K'_{ij}$ by evaluating the kernel function for the training and testing data.  
We present the following
algorithm, which performs the pre-computation and then runs Algorithm~\ref{algAlphatronMain}. 
On the same input, this algorithm produces exactly the same output as Algorithm~\ref{algAlphatron}.

\begin{algorithm}[H]
\label{algAlphatronPre} 
\caption{\textsc{Alphatron\_with\_Pre}}
\SetKw{KwReturn}{Return}
\SetKw{KwInput}{Input}
\SetKw{KwOutput}{Output}
\KwInput{training data $(\mathbf{x}_i, y_i)^{m_1}_{i=1}$, testing data $(\mathbf{a}_i, b_i)^{m_2}_{i=1}$, 
function $u: \mathbb{R} \to [0, 1]$, number of iterations $T$, degree of the multinomial kernel $d$, learning rate $\lambda$.}
\DontPrintSemicolon

\For {$i \leftarrow 1$ \KwTo $m_1$} {
\For {$j \leftarrow 1$ \KwTo $m_1$} {
$K_{ij} \gets \mathcal{K}_d(\mathbf{x}_i, \mathbf{x}_j)$\;
}
}

\For {$i \leftarrow 1$ \KwTo $m_2$} {
\For {$j \leftarrow 1$ \KwTo $m_1$} {
$K'_{ij} \gets \mathcal{K}_d(\mathbf{a}_i, \mathbf{x}_j)$\;
}
}
$\alpha^{t_{out}} \gets$ Run  \textsc{Alphatron\_with\_Kernel} (Algorithm~\ref{algAlphatronMain}) with all inputs as above and $K_{ij}$ and $K'_{ij}$.\;
\KwOutput{$\alpha^{t_{out}}$}
\end{algorithm}

\begin{theorem} [\textsc{Alphatron\_with\_Pre}] \label{thmPre}
Algorithm~\ref{algAlphatronPre} generates the same output as Algorithm~\ref{algAlphatron},
and runs in time $\mathcal{O}(m^2(n+\log d) + Tm^2)$.
\end{theorem}
\begin{proof}
First of all, it is straightforward to see that Algorithm~\ref{algAlphatronPre} behaves in the same way as Algorithm~\ref{algAlphatron},
by using the definition $h^t(\bm{x}) = u\left (\sum_{i=1}^{m_1} \alpha_i^t \mathcal{K}_d(\mathbf{x}, \mathbf{x}_i)\right)$ and noticing that the sequences $\left (\alpha_{\rm Alg\ref{algAlphatron}}^t\right)_{t\in[T]}$ and $\left (\alpha_{\rm Alg\ref{algAlphatronPre}}^t\right)_{t\in[T]}$ are exactly the same. 

For the time complexity, we have $\Ord{m^2}$ inner products to be evaluated.
For each of them we need time $\Ord{n + \log d}$ as we showed in the proof of Theorem~\ref{thmOrig}.
Hence it costs $\mathcal{O}(m^2(n + \log d))$ time to pre-compute
the results of all $\mathcal{K}_d(\bm{x}, \bm{y})$.
By Theorem~\ref{thmVanilla}, the time complexity of \textsc{Alphatron\_with\_Kernel} is $\mathcal{O}(Tm^2)$.
In total the algorithm runs in time $\mathcal{O}(m^2(n+\log d) + Tm^2)$.
\end{proof}
By the pre-computation, we evaluate each kernel function only once with the corresponding memory cost of storing the values.
Comparing with the $\Ord{Tm^2 (n+\log d)}$ time used
for Algorithm~\ref{algAlphatron}, 
Algorithm~\ref{algAlphatronPre} achieves a significant speedup.

\subsection{Classical inner-product estimation}

We can hope to obtain a further speedup by approximating the inner products instead of computing them exactly.
Next, we discuss this inner product approximation, where the approximations rely on sampling data structures, which are discussed in Appendix \ref{appClassical}. These data structures when given a vector allow to sample an index with probability proportional to the components of the vector, as described in Facts \ref{factSamplingl1} and \ref{factSamplingl2}. We call them $\ell_1$ and $\ell_2$ sampling data structures. Here, we use the $\ell_2$ case (Fact \ref{factSamplingl2}), while the second part of this work uses the $\ell_1$ case. 
Based on these data structures, elementary results can be provided to estimate inner products between two vectors. These are described in Lemmas \ref{lemmaSamplingl1} and \ref{lemmaSamplingl2} in Appendix \ref{appClassical}, of which we need Lemma \ref{lemmaSamplingl2} here. 
Our version of the Alphatron algorithm with approximate pre-computation is given in Algorithm \ref{algAlphatronApprox}.
We use the inner product estimation of Lemma \ref{lemmaSamplingl2} 
to improve the run time complexity of Algorithm~\ref{algAlphatronPre}.

\begin{algorithm}[H] 
\label{algAlphatronApprox}
\caption{\textsc{Alphatron\_with\_Approx\_Pre}}
\label{algPreAlpha}
\SetKw{KwReturn}{Return}
\SetKw{KwInput}{Input}
\SetKw{KwOutput}{Output}
\KwInput{training data $(\mathbf{x}_i, y_i)^{m_1}_{i=1}$, testing data $(\mathbf{a}_i, b_i)^{m_2}_{i=1}$, error tolerance parameter
$\epsilon_K$, failure probability $\delta_K$, 
function $u: \mathbb{R} \to [0, 1]$, number of iterations $T$, degree of the multinomial kernel $d$, learning rate $\lambda$}
\DontPrintSemicolon

\For {$i \leftarrow 1$ \KwTo $m_1$} {
Prepare sampling data structure for $\bm{x}_i$ according to Fact~\ref{factSamplingl2}.\;
\For {$j \leftarrow 1$ \KwTo $m_1$} {
$z_{ij} \gets $ Estimate the inner product $ \bm{x}_i \cdot \bm{x}_j$ to $\epsilon_K/(3d)$ additive error
with probability at least $1 - \delta_K/(m_1^2 + m_1 m_2)$  via Lemma \ref{lemmaSamplingl2}.\;
$\widetilde{K}_{i j} \gets \frac{1}{d+1} \sum_{0 \le k \le d}z_{ij}^k$\;
}
}

\For {$i \leftarrow 1$ \KwTo $m_2$} {
Prepare sampling data structure for $\bm{a}_i$ according to Fact~\ref{factSamplingl2}.\;
\For {$j \leftarrow 1$ \KwTo $m_1$} {
$z'_{ij} \gets $ Estimate the inner product $\bm{a}_i \cdot \bm{x}_j$ to $\epsilon_K/(3d)$ additive error
with probability at least $1 - \delta_K/(m_1^2 + m_1 m_2)$ via Lemma \ref{lemmaSamplingl2}.\;
$\widetilde{K}'_{i j} \gets \frac{1}{d+1} \sum_{0 \le k \le d} (z'_{ij})^k$\;
}
}
$\alpha^{t_{out}} \gets$ Call  \textsc{Alphatron\_with\_Kernel} (Algorithm~\ref{algAlphatronMain}) with all input as above and $\widetilde K_{ij}$ and $\widetilde K'_{ij}$.\;
\KwOutput{$\alpha^{t_{out}}$}
\end{algorithm}

\begin{theorem} [Run time of \textsc{Alphatron\_with\_Approx\_Pre}] \label{thmApprox}
Let $\epsilon_K, \delta_K > 0$.
Assume that for all $i \in [m_1]$ and $j \in [m_2]$, $\Vert \bm{x}_i\Vert_2 = \Vert \bm{a}_j \Vert_2 = 1$.
Lines $2-11$ of Algorithm~\ref{algAlphatronApprox} have a run time of 
\[ \tOrd{ mn + \frac{m^2 d^2 }{\epsilon_K^2} \log \frac{1}{\delta_K}
},\]
and provide the kernel matrices $\widetilde K$ and $\widetilde K'$ such that
$\max_{ij} \left \vert \widetilde K_{ij} - K_{ij} \right\vert \leq \epsilon_K$ and $\max_{ij}\left \vert \widetilde K'_{ij} - K'_{ij}\right\vert \leq \epsilon_K$ with success probability $1-\delta_K$.
Line $12$ requires an additional cost of $\Ord{T m^2}$ from the use of Algorithm \ref{algAlphatronMain}.
\end{theorem}
\begin{proof}
For all vectors $\bm{x}_i$ and $\bm{a}_j$, the sampling data structure is prepared in total time $\tOrd {mn}$.
There are $\Ord{m^2}$ inner products to be estimated between these vectors. 
Hence, by Lemma~\ref{lemmaSamplingl2},
each estimation of inner product with additive accuracy $\epsilon_K/ (3d)$ and success probability $1-\delta_K/(m_1^2 + m_1m_2)$
costs $\tOrd {\frac{d^2 }{\epsilon_K^2} \log \frac{m_1^2 + m_1 m_2}{\delta_K}}$ because of $\Vert \bm{x}_i \Vert_2 = \Vert \bm{a}_j \Vert_2 = 1$.
We ignore the $\log (m_1^2 + m_1 m_2)$ factor under the tilde notation compared to $m^2$.
Again, $\Ord{\log d}$ extra time is needed to compute
each multinomial kernel function $\mathcal{K}_d$ from the inner product.
However, we also ignore the $\log d$ factor under the tilde notation.
By Lemma~\ref{lemmaLipschitz} of the Appendix,
the Lipschitz constant for $f(z) = \frac{1}{d+1} \sum_{i=0}^d z^i$
is bounded from above by $3d$, when $z \in [-1, 1]$.
Hence, we obtain 
$\max_{ij} \left\vert \tilde K_{ij} - K_{ij}\right\vert \leq (3d) \cdot \epsilon_K/(3d)$ and $\max_{ij}\left\vert \tilde K'_{ij} - K'_{ij}\right\vert \leq
(3d) \cdot \epsilon_K/(3d)$. 
The last step for calling Algorithm \ref{algAlphatronMain} costs $\Ord{Tm^2}$ again as the matrices are accessible in $\Ord{1}$.
\end{proof}

Since only $m$ sampling state structures are prepared which allow the inner products to be approximated in advance,
Algorithm~\ref{algAlphatronApprox} improves the run time complexity of the Algorithm~\ref{algAlphatronPre}.
However, as the inner products are approximated, we may lose the correctness of
Algorithm~\ref{algAlphatronPre}. 
In Ref.~\cite{pmlr-v99-goel19b}, 
a theoretical upper bound is proven
for the sample complexity
of Algorithm~\ref{algAlphatron}
in the problem setting of Definition \ref{assume:alpha}.
We now show that with approximate pre-computation,
under the same problem and parameter settings as in Ref.~\cite{pmlr-v99-goel19b},
the p-concept error of the output hypothesis does not increase too much.

\begin{theorem} [Correctness of \textsc{Alphatron\_with\_Approx\_Pre}] \label{thm:alpha}
If Definition \ref{assume:alpha} and \ref{assume:param} hold, then
by setting $\delta_K = \delta$, with probability $1 - 3\delta$, Algorithm~\ref{algAlphatronApprox} outputs $\alpha^{t_{\rm out}}$ which describes the hypothesis $h^{t_{\rm out}}(\bm{x}) := u\left(\sum_{i=1}^{m_1} \alpha_i^{t_{\rm out}} \mathcal{K}_d(\mathbf{x}, \mathbf{x}_i)\right)$ such that,
\[
\varepsilon(h) \in \Ord{A_2 + \epsilon_K^2 T^2 + \epsilon_K T },
\]
where $A_2 = L\sqrt{\epsilon} + L \zeta \sqrt[4]{\frac{\log(1/\delta)}{m_1}} + BL\sqrt{\frac{\log(1/\delta)}{m_1}}$.
\end{theorem}
We prove this theorem in Appendix \ref{proof.alphatron_approx}.

\subsection{Quantum Pre-computation}

In the previous subsection, we classically estimate the inner products that construct the kernel matrices. Now, given quantum access to the training data, we replace this estimation with a quantum subroutine and obtain a quantum speedup. 
This section presents the quantum algorithm for pre-computing the kernel matrices used in the Alphatron. 
We assume quantum access to the training data, which includes classical access and also superposition queries to the data.

\begin{defn} [Quantum query access]\label{defQA}
Let $c$ and $n$ be two positive integers
and $\bm u$ be
a vector of bit strings $\mathbf{u} \in (\{0,1\}^{c})^n$. 
Define element-wise quantum access to $\bm u$ for $j\in[n]$ by the operation 
\be
\ket j \ket{0^{c}} \to \ket j \ket{ u_j}, 
\ee
on $\Ord{c+ \log n}$ qubits.
We denote this access by $\bm{QA}(\bm u, n,c)$.
\end{defn}

\begin{oracle}\label{inputQuantumAccess1}
For $k \in [n]$, $i\in [m_1]$, and $j \in [m_2]$,
let $\mathbf{x}_{i}$ and $\mathbf{a}_j$ be the input vectors with $\Vert \bm{x}_i\Vert_2 = \Vert \bm{a}_j\Vert_2 = 1$,
and let $\mathbf{x}_{ik}$ and $\mathbf{a}_{jk}$ be the entries of the vectors.
Assume $c=\Ord{1}$ bits are sufficient to store $\mathbf{x}_{ik}$ and $\mathbf{a}_{jk}$. Assume 
that we are given 
$\bm{QA}(\bm x_i, n,c)$ for each $i \in[m_1]$ and 
$\bm{QA}(\bm a_j, n,c)$ for each $j \in[m_2]$.
\end{oracle}
Our first quantum algorithm is constructed in a straightforward manner. 
We replace the classical approximation of the kernel matrix inner products with a quantum estimation. 
For the quantum estimation of inner products refer to Lemma \ref{lemmaInnerProduct} in Appendix \ref{appQuantum}, which requires quantum query access similar to Data Input \ref{inputQuantumAccess1}. The run time of Lemma \ref{lemmaInnerProduct} depends on the $\ell_2$-norms of the input vectors, which here are $1$.
The result is Algorithm \ref{algAlphatronQuantumPre}.
The run time analysis and the 
guarantees for the output hypothesis are
similar to the classical algorithm. We state them below as a corollary.

\begin{algorithm}[H] 
\label{algAlphatronQuantumPre}
\caption{\textsc{Alphatron\_with\_Q\_Pre}}
\label{algPreAlphaQ}
\SetKw{KwReturn}{Return}
\SetKw{KwInput}{Input}
\SetKw{KwOutput}{Output}
\KwInput{Quantum access to training data $(\mathbf{x}_i, y_i)^m_{i=1}$ and testing data $(\mathbf{a}_i, b_i)^N_{i=1}$ according to Data Input \ref{inputQuantumAccess1}, 
error tolerance parameter
$\epsilon_K$, failure probability $\delta_K$, 
function $u: \mathbb{R} \to [0, 1]$, number of iterations $T$, degree of the multinomial kernel $d$, learning rate $\lambda$}
\DontPrintSemicolon
\;
\For {$i \leftarrow 1$ \KwTo $m_1$} {
\For {$j \leftarrow 1$ \KwTo $m_1$} {
$z_{ij} \gets $ Estimate the inner product $\langle \bm{x}_i, \bm{x}_j \rangle$ to $\epsilon_K/(3d)$ additive error
with probability at least $1 - \delta_K/(m_1^2 + m_1m_2)$  via Lemma \ref{lemmaInnerProduct}.\;
$\widetilde{K}_{ij} \gets \frac{1}{d+1} \sum_{0 \le k \le d}z_{ij}^k$\;
}
}

\For {$i \leftarrow 1$ \KwTo $m_2$} {
\For {$j \leftarrow 1$ \KwTo $m_1$} {
$z'_{ij} \gets $ Estimate the inner product $\langle \bm{a}_i, \bm{x}_j \rangle$ to $\epsilon_K/(3d)$ additive error
with probability at least $1 - \delta_K/(m_1^2 + m_1m_2)$ via Lemma \ref{lemmaInnerProduct}.\;
$\widetilde{K}'_{ij} \gets \frac{1}{d+1} \sum_{0 \le k \le d} (z'_{ij})^k$\;
}
}
$\alpha^{t_{out}} \gets$ Call  \textsc{Alphatron\_with\_Kernel} (Algorithm~\ref{algAlphatronMain}) with all inputs as above and $\widetilde{K}_{ij}$ and $\widetilde{K}'_{ij}$.\;
\KwOutput{$\alpha^{t_{out}}$}
\end{algorithm}

\begin{corollary}[Runtime of \textsc{Alphatron\_with\_Q\_Pre}] \label{cor:generalAlphaQ}
Let $\epsilon_K, \delta_K > 0$.
Assume that for all $i \in [m_1]$ and $j \in [m_2]$, we have quantum query access to the vectors $\bm{x}_i$ and $\bm{a}_j$ via Data Input~\ref{inputQuantumAccess1}.
Lines $2-9$ of Algorithm~\ref{algAlphatronQuantumPre} have a run time of 
\[ \tOrd{ \frac{m^2 d \sqrt{n} }{\epsilon_K} \log \frac{1}{\delta_K} 
}\]
and provide the kernel matrices $\tilde K$ and $\tilde K'$ such that
$\max_{ij} \vert \tilde K_{ij} - K_{ij}\vert \leq \epsilon_K$ and $\max_{ij}\vert \tilde K'_{ij} - K'_{ij} \vert \leq \epsilon_K$ with success probability $1-\delta_K$.
\end{corollary}
\begin{proof}
For $\epsilon_K \in (0,1)$, the run time of each invocation of Lemma \ref{lemmaInnerProduct} is $\tOrd{\frac{d \sqrt n}{\epsilon_K}   \log \left (\frac{m}{\delta_K} \right )  }$, using that the input vectors are in the unit ball. All probabilistic steps in Lines $2-9$ of the algorithm succeed with probability $1-\delta_K$ using a union bound. 
\end{proof}
\begin{corollary}[Guarantee for \textsc{Alphatron\_with\_Q\_Pre}] \label{corQuantumPrecompute}
Let $\delta>0$. 
Assume that for all $i \in [m_1]$ and $j \in [m_2]$, we have quantum query access to the vectors $\bm{x}_i$ and $\bm{a}_j$ via Data Input~\ref{inputQuantumAccess1}.
Let Definitions \ref{assume:alpha} and \ref{assume:param} hold. 
If $A_2 \le 1$, and we set $\epsilon_K = \frac{L\sqrt \epsilon}{T}$ and $\delta_K=\delta$, then 
Algorithm~\ref{algAlphatronQuantumPre} with probability at least $1-3\delta$
outputs $\alpha^{t_{\rm out}}$ which describes the hypothesis $h^{t_{\rm out}}(\bm{x}) := u\left (\sum_{i=1}^{m_1} \alpha_i^{t_{\rm out}} \mathcal{K}_d(\mathbf{x}, \mathbf{x}_i)\right)$ 
such that
\[
\varepsilon\left(h^{t_{\rm out}}\right) \in \Ord{A_2},
\]
with a run time of 
\[ \tOrd{ \frac{ m^2 T d  \sqrt{n} }{L\sqrt \epsilon} \log \frac{1}{\delta} + Tm^2 }.\]
\end{corollary}
\begin{proof}
The proof is analogous to the proof of Theorem \ref{thm:alpha}, where we use Corollary \ref{cor:generalAlphaQ} for the run time of the inner product estimation. 
\end{proof}
\section{Quantum Alphatron} \label{sec:five}

Up to this point, we have been discussing improvements in the pre-computation step of the Alphatron. 
We always use the same \textsc{Alphatron\_with\_Kernel} algorithm once we prepare the kernel matrices $K$ and $K'$.
If data dimension $n$ is much larger than the other parameters, the quantum pre-computation
costs asymptotically more time than \textsc{Alphatron\_with\_Kernel}. Hence, we do not benefit
much from optimizing \textsc{Alphatron\_with\_Kernel} if the cost of preparing the data of size $n$ is taken into account.

However, if we assume that the pre-computation was already done for us, it makes sense to 
discuss quantum speedups for \textsc{Alphatron\_with\_Kernel}, which is what the remainder of this work is about. 
In other words, we consider the following scenario.
\begin{oracle} \label{inputQueryKernel}
Let there be given two training data sets $(\mathbf{x}_i, y_i)^{m_1}_{i=1} \in \mathbb{B}^{n} \times [0,1]$ and $(\mathbf{a}_i, b_i)^{m_2}_{i=1} \in \mathbb{B}^{n} \times [0,1]$, which define the kernel matrices
$K_{ij} := \mathcal{K}_d(\mathbf{x}_i, \mathbf{x}_j)$ 
and
$K'_{ij} := \mathcal{K}_d(\mathbf{a}_i, \mathbf{x}_j)$. Let each entry $K_{ji}$ and $K'_{ji}$ be specified by $\Ord{1}$ bits.
We assume that
we have query access to each entry in $\Ord{1}$.
\end{oracle}
The bottleneck
of the computation in the \textsc{Alphatron\_with\_Kernel} is the cost of about $\Ord{Tm}$ for the inner product evaluations. By the sampling
techniques and quantum estimation, we may speed them up.

\subsection{Main loop with approximated inner products}

We employ the classical sampling of inner products in the \textsc{Alphatron\_with\_Kernel} algorithm. The result is Algorithm~\ref{algAlphatronSamp}.
For the kernel matrices $K_{ji}$ and $K'_{ji}$, define $K_{\max}$ as an upper bound for $|K_{ji}|$ and $|K'_{ji}|$.

\begin{algorithm}[H] 
\label{algAlphatronSamp}
\caption{\textsc{Alphatron\_with\_Kernel\_and\_Sampling}}
\SetKw{KwReturn}{Return}
\SetKw{KwInput}{Input}
\SetKw{KwOutput}{Output}
\KwInput{
training data $(\mathbf{x}_i, y_i)^{m_1}_{i=1}$, testing data $(\mathbf{a}_i, b_i)^{m_2}_{i=1}$, error parameter
$\epsilon_I$ and failure probability $\delta$, 
function $u: \mathbb{R} \to [0, 1]$, number of iterations $T$, degree of the multinomial kernel $d$, learning rate $\lambda$, query access to $K_{ji}$ and $K'_{ji}$, the upper bound $K_{\max}$
for both $|K_{ij}|$ and $|K'_{ij}|$}
\DontPrintSemicolon

$\alpha^0 \leftarrow \mathbf{0} \in \mathbb{R}^{m_1}$\;
\For {$t \leftarrow 0$ \KwTo $T-1$} {
Prepare sampling data structure for $\alpha^t$ via Fact~\ref{factSamplingl1}\;
\For {$j \leftarrow 1$ \KwTo $m_1$} {
Define $K_j$ as the vector $(K_{j1}, K_{j2}, \cdots, K_{jm_1})$ \;
$r^t_j \gets$ Estimate inner product $\alpha^t \cdot K_{j}$ to additive accuracy $\epsilon_I$
with success probability $1-\delta/(2Tm_1)$ via Lemma \ref{lemmaSamplingl1}\;
$\alpha_j^{t+1} \leftarrow \alpha_j^t + \frac{\lambda}{m_1} (y_j - u(r^t_j) )$\;
}
\For {$j \leftarrow 1$ \KwTo $m_2$} {
Define $K'_j$ as the vector $(K'_{j1}, K'_{j2}, \cdots, K'_{j m_1})$\;
$s^t_j \gets$ Estimate inner product $\alpha^t \cdot K'_{j}$ to additive accuracy $\epsilon_I$ with success probability $1-\delta/(2Tm_2)$ via Lemma \ref{lemmaSamplingl1}\;
}
}
$t_{\rm out} \gets \arg \min_{t \in [T]} \frac{1}{m_2} \sum_{j=1}^{m_2} (u(s^t_j) - b_j)^2$\;
\KwOutput{$\alpha^{t_{\rm out}}$} \label{AMSvalid}\;
\end{algorithm}

\begin{theorem} \label{thm:alphaApproxLoop}
We assume query access Data Input \ref{inputQueryKernel} to the kernel matrices $K$ and $K'$ with known $K_{\max}$.
Let $\epsilon_I,\delta \in (0,1)$.
If the Definitions \ref{assume:alpha} and \ref{assume:param} hold, the Algorithm \ref{algAlphatronSamp} outputs $\alpha^{t_{\rm out}}$ which describes the hypothesis $h^{t_{\rm out}}(\bm{x}) := u\left (\sum_{i=1}^{m_1} \alpha_i^{t_{\rm out}} \mathcal{K}_d(\mathbf{x}, \mathbf{x}_i)\right)$ such that with probability $1 - 3\delta$,
\[
\varepsilon(h^{t_{\rm out}}) \in \Ord{A_2 + L \epsilon_I + L^2 \epsilon_I^2 },
\]
where $A_2$ is defined in Definition~\ref{assume:param}. 
The run time of this algorithm is $\tOrd{Tm + T^3 m \frac{ K_{\max}^2 }{ L^2 \epsilon_I^2} \log \left( \frac{1}{\delta}\right) }$.
Moreover, if $A_2 \le 1$, and we set $\epsilon_I = \sqrt \epsilon$, then we obtain the guarantee
\[
\varepsilon(h^{t_{\rm out}}) \in \Ord{A_2},
\]
and have a run time of $\tOrd{Tm + T^3 m \frac{K_{\max}^2 }{L^2 \epsilon} \log \left( \frac{1}{\delta}\right) }$.
\end{theorem}
\begin{proof}
By the definition
of $r^t_j$, we have $| r^t_j - \sum_{i = 1}^{m_1} \alpha^t_i K_{ji} | \le \epsilon_I$,
and by the definition of $s^t_j$, we have $| s^t_j - \sum_{i = 1}^{m_1} \alpha^t_i K'_{ji} | \le \epsilon_I$.
By Definition~\ref{def:hypofunc} and by the Lipschitz condition of $u$,
we obtain that $\left\vert u(r^t_j) -  g(\alpha^t, K, j)\right\vert\le L \epsilon_I$,
and $\left\vert u(s^t_j) -  g(\alpha^t, K', j) \right\vert \le L \epsilon_I$.

Consider the cases in Eqns.~(\ref{eqConvergenceCase1}) and (\ref{eqConvergenceCase2}) in the proof of Theorem~\ref{thm:alpha} for the sequence of $\tilde \alpha^t$ generated by Algorithm \ref{algAlphatronSamp}. Similarily, there exists $t^*$ such that Case 2 holds.
Then by Lemma~\ref{lemConvergenceModified} with $\omega = \alpha^{t^*}$,
we obtain
\be
\Vert \textbf{v}(\tilde \alpha\super{t^*}) - \Hyp{v}\Vert_2^2 - \Vert\textbf{v}(\tilde\alpha\super{t^*+1}) - \Hyp{v}\Vert_2^2
\geq \frac{1}{L^2}\hat{\varepsilon}(h^{t^\ast})- A_{1} - \epsilon_I^2.
\ee
Hence it now holds that
\be
\frac{B \eta}{L} \geq \frac{1}{L^2}\hat{\varepsilon}(h^{t^\ast})- A_{1} - \epsilon_I^2,
\ee
which implies that 
\be
\hat \varepsilon(h^{t^\ast}) \leq B L \eta + L^2 A_{1} + L^2 \epsilon_I^2.
\ee
Using the known bound for $\eta$ we have 
\be
\hat{\varepsilon}(h^{t^*}) \in \Ord{A_2+
L^2 \epsilon_I^2 }.
\ee
Again, by the Rademacher analysis in proof of Theorem~\ref{thm:alpha},
we obtain
\be
\varepsilon(h^{t^*}) \leq \hat{\varepsilon}(h^{t^*}) + \Ord{BL \sqrt{\frac{1}{m_1}}
+  \sqrt{\frac{\log (1/\delta)}{m_1}}} \in \Ord{A_2 + L^2 \epsilon_I^2}.\label{eq:cd1}
\ee
We define $t' := \arg\min_{t \in [T]} \varepsilon(h^t)$.
Then $\varepsilon(h^{t'}) \le \varepsilon(h^{t^*})$.
By Lemma~\ref{lemma:ValidationError} with $\Gamma_j^t = u(s_j^t)$, at Line~\ref{AMSvalid} in
Algorithm~\ref{algAlphatronSamp}, we obtain $h^{\tilde t}$ such that
\be
| \hat{\err}(h^{\tilde t}) - \hat{\err}(h^{t'}) | \le L \epsilon_I. \label{eq:cd2}
\ee
As in the proof of the Theorem~\ref{thm:alpha}, by Chernoff bound,
setting $m_2 \in \Ord{m_1 \log (T / \delta) }$, with probability $1 - \delta$,
we have
\be
\forall t \in [T], | \err(h^t) - \hat{\err}(h^t) | \in \Ord{A_2}.\label{eq:cd3}
\ee
Using the same idea as
in the last part of the proof of the Theorem~\ref{thm:alpha},
we relate above inequalities (\ref{eq:cd1}), (\ref{eq:cd2}), and (\ref{eq:cd3}), and obtain that for the output hypothesis $h^{\tilde t}$,
\be
\varepsilon\left(h^{\tilde t}\right) \in \Ord{A_2 + L \epsilon_I + L^2 \epsilon_I^2 }.
\ee
For the run time complexity, the total time of preparing the sampling data structure for $\alpha^t$
is $\tOrd{Tm}$ because we prepare $\Ord{T}$ such structures and preparing each of them costs $\tOrd{m}$.
By Lemma~\ref{lem:alphabound}, we have the upper bound $
\max_t \Vert \alpha^t \Vert_1 \le  \frac{T}{L}$.
Hence, the run time of each estimation $r^t_j$ and $s^t_j$ via Lemma~\ref{lemmaSamplingl1} is bounded by
$\tOrd{\frac{T^2 K_{\max}^2 }{ L^2\epsilon_I^2}\log {1 /\delta}}$.
Then the total run time is 
\be
\tOrd{Tm + T^3 m \frac{ K_{\max}^2 }{L^2 \epsilon_I^2} \log \left( \frac{1}{\delta}\right) }.
\ee
Setting $\epsilon_I= \sqrt \epsilon$ obtains $\epsilon(h) \in \Ord{A_2}$ with the run stated in the theorem. 
\end{proof}
Now, replace the classical sampling of the inner product with the quantum estimation of the inner product.
With the Lemma \ref{lemmaInnerUnitary} in Appendix \ref{appQuantum}, 
we can remove the explicit dimension dependence of the inner product estimation, at the expense of using a QRAM, see Definition \ref{defn:QRAM} in the next section. 

\subsection{Quantum speedup for the main loop}

For the quantum algorithm, we assume the quantum query access to the kernel matrices $K$ and $K$.
Note the definition of quantum query access in Definition \ref{defQA} in Appendix \ref{appQuantum}.
\begin{oracle}\label{input:quantumkernel}
Assume Data Input \ref{inputQueryKernel} for the training data and the kernel matrices. 
For all $j \in [m_1]$,
define $K_j$ as the vector $(K_{j1}, K_{j2}, \cdots, K_{jm_1})$, and 
for all $j \in [m_2]$,
define $K'_j$ as the vector $(K'_{j1}, K'_{j2}, \cdots, K'_{jm_1})$. 
Assume the availability of the quantum access $\bm{QA}(K_{j},m_1, \Ord{1})$, for all $j \in [m_1]$, and 
the quantum access $\bm{QA}(K'_{j},m_1, \Ord{1})$ , for all $j \in [m_2]$. \end{oracle}
Based on this input a simple circuit prepares query access to the non-negative versions of the vectors.
\begin{lemma} \label{lemKernelPosNeg}
Assume Data Input \ref{inputQueryKernel} and 
define the non-negative vectors $(K_j)^+$, $(K_j)^-$, with $K_j =(K_j)^+ - (K_j)^-$ and 
the non-negative vectors $(K'_j)^+$, $(K'_j)^-$, with $K'_j =(K'_j)^+ - (K'_j)^-$.
Given Data Input \ref{input:quantumkernel}, then query accesses  $\bm{QA}((K_{j})^+,m_1, \Ord{1}), \bm{QA}((K_{j})^-,m_1, \Ord{1}), \forall j\in[m_1]$ and query accesses
$\bm{QA}((K'_{j})^+,m_1, \Ord{1}), \bm{QA}((K'_{j})^-,m_1, \Ord{1}), \forall j\in[m_2]$
can be provided with two queries to the respective inputs and a constant depth circuit of quantum gates. 
\end{lemma}
For our quantum version for the main loop of the \textsc{Alphatron} algorithm, we will also require a dynamic quantum data structure for the $\alpha$ vector which allows us to obtain efficient quantum sample access. 
\begin{defn} [Quantum sample access]\label{defQS}
Let $c_1$, $c_2$, and $n$ be positive integers
and $\bm v'\in (\{0,1\}^{c_1})^n$ and $\bm v''\in (\{0,1\}^{c_2})^n$ be
vectors of bit strings. 
Define quantum sample access to a vector $\bm v$ via the operation 
\be
\ket{\bar 0} \to \frac{1}{\sqrt{\Vert \mathcal Q(\bm v', \bm v'')\Vert_1 }} \sum_{j=1}^n  \sqrt{\mathcal Q(v'_j,v''_j)} \ket j,
\ee 
on $\Ord{\log n}$ qubits. 
We denote this access by $\bm{QS}(\bm v, n,c_1,c_2)$. For the sample access to a vector $\bm v$ which approximates a vector with components in $[0,1]$, we use the shorthand notation $\bm{QS}(\bm v, n, c_2) := \bm{QS}(\bm v, n,1,c_2)$.
\end{defn}
One way to obtain such an access is via quantum random access memory (QRAM) \cite{Giovannetti2008,Giovannetti2008_2, Arunachalam2015}. Such a device stores the data 
in (classical) memory cells, but 
allows for superposition queries to the data. If all the partial sums 
are also stored, then QRAM can provide quantum sample access via the Grover-Rudolph procedure, see Ref.~\cite{GR2002}.
This costs resources proportional to the length of the vector to set up, but then can provide the the superposition state in a run time logarithmic in the length of the vector. 

\begin{defn}[Quantum RAM] 
\label{defn:QRAM}
Let $c$ and $m$ be positive integers.
Let $\bm v$ be a vector of dimension $m$, where each element of $v$ is a bit string of length $c$, i.e., $v \in \mathbb (\{0,1\}^{c})^{m}$. Quantum RAM is defined such that  with a one-time cost of $\tOrd{c\ m}$ we can construct quantum query and sampling access
$\bm {QA}(v, m, c)$ and $\bm {QS}(v, m, c)$, see Definitions \ref{defQA} and \ref{defQS} in Appendix \ref{appQuantum}. 
Each query costs $\Ord{c\ {\rm poly} \log m}$.
\end{defn}
Based on Data Input \ref{input:quantumkernel} and Definition~\ref{defn:QRAM}, Lemma~\ref{lemmaInnerUnitary} in Appendix \ref{appQuantum} allows us to estimate the
inner products between $\alpha^t$ and $K_j$ more efficiently than the equivalent estimation in Algorithm~\ref{algAlphatronSamp}.
We have the following algorithm.

\begin{algorithm}[H] 
\label{algQAlphatron2}
\caption{\textsc{Quantum\_Alphatron}}
\SetKw{KwReturn}{Return}
\SetKw{KwInput}{Input}
\SetKw{KwOutput}{Output}
\KwInput{training data $(\mathbf{x}_i, y_i)^{m_1}_{i=1}$, testing data $(\mathbf{a}_j, b_j)^{m_2}_{j=1}$, error tolerance parameters
$\epsilon_I$ and $\delta_I$, 
function $u: \mathbb{R} \to [0, 1]$, number of iterations $T$, degree of the multinomial kernel $d$, learning rate $\lambda$, quantum query access to $K_j, \forall j\in[m_1]$ and $K'_j, \forall j\in[m_2]$ via Data Input~\ref{input:quantumkernel} }
\DontPrintSemicolon

$\alpha^0 \leftarrow \mathbf{0} \in \mathbb{R}^{m_1}$\;

\label{line:qmf1} \For {$j \leftarrow 1$ \KwTo $m_1$} {
$p^{\max}_j \gets \max_i \vert K_{ji}\vert$ via quantum maximum finding with success probability $1 - \delta_I / (4m_1)$\;
Define the non-negative vectors $(K_j)^+$, $(K_j)^-$, with $K_j =(K_j)^+ - (K_j)^-$\;
From query access to $K_j$ provide query access to $(K_j')^+$, $(K_j')^-$ via Lemma \ref{lemKernelPosNeg}\;
}
\For {$j \leftarrow 1$ \KwTo $m_2$} {
$q^{\max}_j \gets \max_i \vert K'_{ji}\vert$ via quantum maximum finding with success probability $1 - \delta_I / (4m_2)$\;
Define the non-negative vectors $(K_j')^+$, $(K_j')^-$, with $K_j =(K_j')^+ - (K_j')^-$\;
From query access to $K_j$ provide query access to $(K_j')^+$, $(K_j')^-$ via Lemma \ref{lemKernelPosNeg}\;
}\label{line:qmf2}
\For {$t \leftarrow 0$ \KwTo $T-1$} {
Store in QRAM (see Definition~\ref{defn:QRAM}) 
the non-negative vectors $(\alpha^t)^+,(\alpha^t)^-$, where $\alpha^t = (\alpha^t)^+ - (\alpha^t)^-$, where each element of the vector is stored using $\left \lceil \log (\lambda T/m_1) \right \rceil + \left \lceil \log\left(\frac{2K_{\max} m_1}{\epsilon_I} \right) \right\rceil $ bits \label{line:qram}\;
$w_t \gets \Vert \alpha^t \Vert_1$\;
\For {$j \leftarrow 1$ \KwTo $m_1$} {
$r_j^t \gets$ 
Estimate inner product $\alpha^t \cdot K_j$, by estimating  $(\alpha^t)^+ \cdot (K_{j})^+$, $(\alpha^t)^+ \cdot (K_{j})^-$, $(\alpha^t)^- \cdot (K_{j})^+$, and $(\alpha^t)^- \cdot (K_{j})^-$ via Statement (iii) of Lemma~\ref{lemmaInnerUnitary} (using $w_t$ and $p^{\max}_j$),
each to additive accuracy $\epsilon_I/8$ with success probability $1 - \delta_I / (16Tm_1)$\;
$\alpha_j^{t+1} \leftarrow \alpha_j^t + \frac{\lambda}{m_1} (y_j - u(r_j^t) )$\;	
}
\For {$j \leftarrow 1$ \KwTo $m_2$} {
$s^t_j \gets$ Estimate inner product $\alpha^t \cdot K'_j$, by estimating $(\alpha^t)^+ \cdot (K'_{j})^+$, $(\alpha^t)^+ \cdot (K'_{j})^-$, $(\alpha^t)^- \cdot (K'_{j})^+$, and $(\alpha^t)^- \cdot (K'_{j})^-$ via Statement (iii) of Lemma~\ref{lemmaInnerUnitary} (using $w_t$ and $q^{\max}_j$),
each to additive accuracy $\epsilon_I/8$ with success probability $1 - \delta_I / (16Tm_2)$\;
}
}
$t_{out}\leftarrow \arg\min_{t\in [T]}\frac{1}{m_2}\sum_{j=1}^{m_2} (u(s_j^t)-b_j)^2$ \;
\KwOutput{$\alpha^{t_{out}}$ }\;
\end{algorithm}

\begin{theorem} [Quantum Alphatron] \label{thm:alphaQuantumLoop}
We assume quantum query access to the vectors $K_j$ and $K'_j$ via Data Input~\ref{input:quantumkernel}.
Again, let $K_{\max}$ be maximum of all 
entries in $K$ and $K'$.
Let $\delta \in (0,1)$.
Given Definitions \ref{assume:alpha} and \ref{assume:param} and $\delta_I = \delta$, Algorithm \ref{algQAlphatron2} 
outputs $\alpha^{t_{out}}$ such that the hypothesis $h^{t_{\rm out}} = u\left(\sum_j\alpha_j^{t_{\rm out}} \psi(\bm x_j) \cdot \psi(\bm x)\right)$ satisfies with probability $1 - 3\delta$,
\[
\varepsilon(h^{t_{\rm out}}) \in \Ord{A_2+L\epsilon_I +L^2 \epsilon_I^2},
\]
where $A_2$ is defined in Theorem~\ref{thm:alpha}. The run time of this algorithm is 
\[\tOrd{m^{1.5}\log \left( \frac{1}{\delta} \right) + Tm + T^2 m \frac{ K_{\max}}{L \epsilon_I} \log \left( \frac{1}{\delta}\right) }.\]
If $A_2 \le 1$ and we set $\epsilon_I = \sqrt \epsilon$ then we further obtain the guarantee
\[
\varepsilon(h^{t_{\rm out}}) \in \Ord{A_2},
\]
and the run time is 
\[ \tOrd{m^{1.5}\log \left( \frac{1}{\delta} \right) + Tm + T^2 m \frac{ K_{\max} }{L \sqrt \epsilon} \log \left( \frac{1}{\delta}\right) }.\]
\end{theorem}
\begin{proof}
First, consider the numerical error from truncating the $\alpha$ vectors. 
Recall that in the classical steps of the algorithm we work in the arithmetic model where all the steps occur at infinite precision. 
Let $\alpha^t\in \mathbb R^{m_1}$ be the vector given to infinite precision (arithmetic model) with known $0 <  \alpha_{\max} \leq \lambda T/ m_1$ (Lemma \ref{lem:alphabound}). Set $c_1 \geq \lceil \log (\lambda T/m_1) \rceil$ and
$c_2 \geq \left \lceil \log\left(\frac{2K_{\max} m_1}{\epsilon_I} \right) \right \rceil$.
Let $\tilde \alpha \in \mathbb \{0,1\}^{c_1+c_2}$ be  the element-wise $c_1 + c_2$ bit approximation of $\alpha^t$ (stored in QRAM).
Note that 
\be \label{eqqnum}
\vert \alpha^t \cdot K_i - \tilde \alpha^t \cdot K_i \vert \leq m_1 K_{\max} \max_{j\in [m_1]} \vert \alpha^t_j -\tilde \alpha^t_j\vert \leq   \frac{m_1 K_{\max}}{2^{c_2+1}} \leq \frac{\epsilon_I}{2}.
\ee
Aside from the estimation of the inner products,
the remaining part of Algorithm~\ref{algQAlphatron2} is the same
as Algorithm~\ref{algAlphatronSamp}.
Compared to Algorithm~\ref{algAlphatronSamp}, we change the accuracy of the inner product estimation to $\epsilon_I/2$, 
hence we achieve that
\be
\vert r_j^t - \tilde \alpha^t \cdot K_j\vert \leq \frac{\epsilon_I}{2}, \forall t\in[T], \forall j\in[m_1],\\
\vert s_j^t - \tilde \alpha^t \cdot K'_j\vert \leq \frac{\epsilon_I}{2}, \forall t\in[T], \forall j\in[m_2],
\ee
with the stated success probabilities.
Using Eq.~(\ref{eqqnum}) for the numerical error,
we obtain that
\be
\vert r_j^t - \alpha^t \cdot K_j\vert \leq \epsilon_I, \forall t\in[T], \forall j\in[m_1],\\
\vert s_j^t - \alpha^t \cdot K'_j\vert \leq \epsilon_I, \forall t\in[T], \forall j\in[m_2],
\ee
with the same success probabilities.
Hence the same accuracy guarantees holds as in the proof of Theorem~\ref{thm:alphaApproxLoop}.
For the output hypothesis $h^{t_{\rm out}}$, we have
\be
\varepsilon(h^{t_{\rm out}}) \in \Ord{A_2 + L \epsilon_I +
L^2 \epsilon_I^2}.
\ee
For the run time complexity, there are three terms.
From Line~\ref{line:qmf1} to Line~\ref{line:qmf2}, we perform $\Ord{m_1+m_2} = \Ord{m}$ 
quantum maximum findings. 
The run time of a single run of the quantum maximum finding is bounded by $\tOrd{\sqrt{m} \log (1/\delta)}$ \cite{DH96}. Hence, this part of the algorithm takes $\Ord{m^{1.5} \log (1/\delta)}$ time.
In Line~\ref{line:qram},
the time of storing all $\alpha^t$ in QRAM is $\tOrd{Tm}$ because we have $\Ord{T}$ vectors and 
storing each of them costs $\tOrd{m}$ time.
For each step $t$, the run time of the estimations
$r^t_j$ and $s^t_j$ depends on the norm $\Vert \alpha^t \Vert_1 \leq T/L$. 
Hence, the run time of each estimation $r^t_j$ and $s^t_j$ via (iii) of Lemma~\ref{lemmaInnerUnitary} is bounded by
$\tOrd{\frac{ T K_{\max} }{ L\epsilon_I} \log \left( \frac{1}{\delta}\right)}$,
and we need to estimate $\Ord{Tm}$ inner products.
Then the overall run time is 
\be
\tOrd{m^{1.5}\log \left( \frac{1}{\delta} \right) + Tm + T^2 m \frac{ K_{\max}}{L \epsilon_I} \log \left( \frac{1}{\delta}\right) }.
\ee
If we set $\epsilon_I= \sqrt \epsilon$, we obtain $\epsilon(h) \le O(A_2)$ and the run time is
\be
\tOrd{m^{1.5}\log \left( \frac{1}{\delta} \right) + Tm + T^2 m \frac{ K_{\max} }{L\sqrt\epsilon} \log \left( \frac{1}{\delta}\right) }.
\ee
\end{proof}

\section{Discussion}
\label{sec:discuss}

In this section, we summarize the results from previous sections and discuss the improvement on the run time complexity by pre-computation and quantum estimation.
In Section~\ref{sec:four}, we have introduced \textsc{Alphatron\_with\_Pre},\textsc{Alphatron\_with\_Approx\_Pre}, and \textsc{Alphatron\_with\_Q\_Pre}, which improve the
original \textsc{Alphatron}. The scenario is that the dimension of the data $n$ is  much larger than 
the other parameters, a situation that is relevant for many practical applications.
Without the pre-computation, we have a run time $\Ord{Tm^2 n}$ compared to the run time with the pre-computation of $\Ord{ m^2 n + m^2 \log d + Tm^2}$. The factor of the $n$ dependent term loses a factor $T$ which is a small improvement.
Moreover, by quantum amplitude estimation, we gain a quadratic speedup in the dimension $n$.
We list the results of Section~\ref{sec:four} in Table~\ref{tablePreGeneral} for comparison.

\begin{table}[H]
\begin{center}
\begin{tabular} { |c | c | c | c | c }
\hline
Name & Pre-computation & Main loop & Proved in \\
\hline \hline
\textsc{Alphatron} &  not applicable
& $\Ord{Tm^2n}$ &\cite{pmlr-v99-goel19b}, also Theorem~\ref{thmOrig} \\
\hline
\textsc{Alphatron\_with\_Pre} & $\Ord{ m^2 n + m^2 \log d}$ & $\Ord{Tm^2}$ & Theorem~\ref{thmPre}\\
\hline
\textsc{Alphatron\_with\_Approx\_Pre} & $\tOrd{mn + \frac{m^2 d^2 }{\epsilon_K^2} \log \frac{1}{\delta_K}}$ & $\Ord{Tm^2}$ & Theorem~\ref{thmApprox} 
\\
\hline
\textsc{Alphatron\_with\_Q\_Pre} & $\tOrd{ \frac{m^2 d \sqrt{n} }{\epsilon_K} \log \frac{1}{\delta_K} }$ & $\Ord{Tm^2}$ & Corollary~\ref{cor:generalAlphaQ}\\
\hline
\end{tabular}
\caption{Comparison of the first set of algorithms in  Section~\ref{sec:four}. 
We separate the pre-computation of the multinomial kernel function from the main loop and also estimate the training set inner products instead of computing them exactly, which can improve the time complexity of the computation of the kernel function.
For all algorithms, we indicate the general result without using the learning setting in Definition \ref{assume:alpha}. For  \textsc{Alphatron\_with\_Approx\_Pre} and
\textsc{Alphatron\_with\_Q\_Pre}, the relevant kernel functions are estimated to accuracy
$\epsilon_K$ with failure probability $\delta_K$. 
To obtain the weak p-concept learning result of Theorem \ref{thmLearning} for all these algorithms, take the concept class defined in Definition \ref{assume:alpha} and the parameter settings for the algorithms of Definition \ref{assume:param}. Also, set $\epsilon_K = L\sqrt \epsilon / T$ and $\delta_K = \delta$. We do not further evaluate the formulas (using, e.g., the expressions for $T$ and $m_1$) as the main focus of this table is on the  dependency on $n$
which dominates all other parameters.} 
\label{tablePreGeneral}
\end{center}
\end{table}

Section~\ref{sec:five} has introduced \textsc{Alphatron\_with\_Kernel\_and\_Sampling} and
\textsc{Quantum\_} \textsc{Alphatron}.
In this scenario, we assume constant time query access to the kernel matrices (i.e., to the result of the pre-computation), while
the quantum version requires quantum query access.
These algorithms only focus on the main loop part of the \textsc{Alphatron}.
Hence, these algorithms can be viewed as the improvements for \textsc{Alphatron\_with\_Kernel}.
We list the results of Section~\ref{sec:five} in Table~\ref{tableOracleGeneral}.
For comparison, we also list the time complexity of \textsc{Alphatron\_with\_Kernel}.

\begin{table}[t]
\begin{center}
\begin{tabular} { |c | c | c | c | c }
\hline
Name &  Main loop & Theorem\\
\hline \hline
\textsc{Alphatron\_with\_Kernel} &  $\Ord{Tm^2}$ &Theorem \ref{thmVanilla} \\
\hline
\textsc{Atron\_with\_Kernel\_and\_Sampling} &  $\tOrd{Tm + T^3 m \frac{  K_{\max}^2 }{ L^2\epsilon_I^2}\log \frac{1}{\delta}  }$ &Theorem \ref{thm:alphaApproxLoop} \\
\hline
\textsc{Quantum\_Alphatron} &  $\tOrd{m^{1.5}\log\frac{1}{\delta} + Tm + T^2 m \frac{  K_{\max} }{L \epsilon_I} \log \frac{1}{\delta} }$ & Theorem \ref{thm:alphaQuantumLoop} \\
\hline
\end{tabular}
\caption{Comparison of the second set of algorithms \textsc{Atron\_with\_Kernel\_and\_Sampling} and
\textsc{Quantum\_Alphatron}, which are discussed in Section~\ref{sec:five}, to \textsc{Alphatron\_with\_Kernel}. These algorithms change the main loop part by using an inner product estimation.
The inner product estimation is performed to accuracy $\epsilon_I$ and the total success probability of the algorithm is $1-\delta$.
Here, we indicate the general result without the learning setting in Definition \ref{assume:alpha}.
}
\label{tableOracleGeneral}
\end{center}
\end{table}

For Table~\ref{tableOracleGeneral}, it is not
obvious that the quantum algorithm has a speedup compared to
the \textsc{Alphatron\_with\_Kernel}.
As mentioned in the preliminary, $\mathcal{K}_d (\bm{x}, \bm{y}) \le 1$ for all
$\bm{x}, \bm{y} \in \mathbb{B}^{n}$. Hence, 
we can use $K_{\max} \leq 1$.
Recall from Theorem \ref{thmLearning} 
that if $m_1 \geq \max \left \{m'_1, m''_1 \right \}$, with
$m'_1 = \frac{16 \zeta^4}{ \epsilon^2} \log(1/\delta)$ and $m''_1 = \frac{ 4 B^2}{ \epsilon} \log(1/\delta)$,
then the concept class in Definition \ref{assume:alpha} is weak p-concept learnable up to $2 C^{\prime\prime} L \sqrt \epsilon$ by the \textsc{Alphatron} algorithm.

\textbf{Case $m'_1>m''_1$.} Consider the first case, which is equivalent to
$4\zeta^4 > B^2 \epsilon$. Hence, $m_1 \in \Ord{m'_1}$ leads to learnability, which we can use to simplify
$T = C B L \sqrt{\frac{m_1}{\log (1/\delta)}} \in \Ord{ B L \frac{ \zeta^2}{ \epsilon}}
$.
In addition, $m=m_1+m_2$, 
and we have $m_2 = C^\prime m_1\log(T/\delta)$, hence,
$
m= (1+ C' \log T + C' \log 1/\delta ) m_1 \in \tOrd{ \frac{ \zeta^4}{ \epsilon^2} \log^2 \frac{1}{\delta}}.
$
From Theorem \ref{thm:alphaQuantumLoop}, we have the run time 
\be
\tOrd{m^{1.5}\log\frac{1}{\delta} + Tm + T^2 m \frac{1 }{L\sqrt \epsilon}\log \frac{1}{\delta}  } &=& \tOrd{  \frac{\zeta^6}{\epsilon^3} \log^4 \frac{1}{\delta} + B L \frac{\zeta^4}{ \epsilon^3}  \log^2 \frac{1}{\delta} +  B^2 L \frac{\zeta^8}{ \epsilon^{4.5}} \log^3 \frac{1}{\delta}} \nonumber
\\ &=&\tOrd{  B^2 L \frac{\zeta^8}{ \epsilon^{4.5}} \log^4 \frac{1}{\delta}}. 
\ee
For the classical run time, we simplify $\Ord{ T m^2 } \subseteq \tOrd{B L \frac{\zeta^{10}}{ \epsilon^5} \log^4 \frac{1}{\delta}}$.

\textbf{Case $m''_1>m'_1$.} Consider the second case, which is equivalent to
$4\zeta^4 < B^2 \epsilon$. Hence, $m_1 = \Ord{ m''_1}$ leads to learnability, which we can use to simplify
$T = C B L \sqrt{\frac{m_1}{\log (1/\delta)}} \in \Ord{  \frac{B^2 L }{\sqrt \epsilon}}$.
In addition, 
$m \in  \tOrd{ \frac{ B^2}{ \epsilon} \log^2 \frac{1}{\delta}}$.
From Theorem \ref{thm:alphaQuantumLoop}, we have the run time $\tOrd{\frac{B^6 L }{ \epsilon^{2.5}}   \log^4 \frac{1}{\delta}}$. 
For the classical run time we simplify
$\Ord{Tm^2} \subseteq \tOrd{\frac{  B^6 L}{ \epsilon^{2.5}} \log^4\frac{1}{\delta}}.
$
\begin{table}[t]
\begin{center}
\begin{tabular} { |c | c | c | c | c }
\hline
Case &  Classical run time & Quantum run time & Advantage \\
\hline \hline
$4\zeta^4 > B^2 \epsilon$ &  $\tOrd{ \frac{B L \zeta^{10}}{ \epsilon^5} \log^4 \frac{1}{\delta}}$ & 
$\tOrd{ \frac{B^2 L \zeta^8}{ \epsilon^{4.5}} \log^4 \frac{1}{\delta}}$ & yes
\\
\hline
$4\zeta^4 \leq B^2 \epsilon$ &  $\tOrd{\frac{  B^6 L}{ \epsilon^{2.5}} \log^4\frac{1}{\delta}}$ & $\tOrd{\frac{B^6 L }{ \epsilon^{2.5}}   \log^4 \frac{1}{\delta}}$
& no
\\
\hline
\end{tabular}
\caption{Comparison of the algorithms \textsc{Alphatron\_with\_Kernel} and \textsc{Quantum\_Alphatron} for the learning setting in Definition \ref{assume:alpha}. Only in the first case a quantum advantage is obtained.
}
\label{tableOracleAdvantage}
\end{center}
\end{table}
This analysis of the two cases is summarized in Table \ref{tableOracleAdvantage} and allows us to state our final theorem. 
\begin{theorem}[Quantum p-concept learnability via the \textsc{Quantum\_Alphatron}]
\label{thmQuantumLearning}
Let the concept class and distribution be defined by Definition \ref{assume:alpha}. For this concept class, let $4\zeta^4 > B^2 \epsilon$.
In addition, let there be given quantum access to the kernel matrices via Data Input~\ref{input:quantumkernel}.
Then, 
the concept class in Definition \ref{assume:alpha} is weak p-concept learnable up to $2 C^{\prime\prime} L \sqrt \epsilon$ by the \textsc{Quantum\_Alphatron} algorithm with a run time that shows an advantage by a factor $\sim \frac{\zeta^2}{B \sqrt \epsilon}$ over the classical algorithm given the same input.
\end{theorem}
A note on the condition $4\zeta^4 > B^2 \epsilon$ for the speedup. By Definition \ref{assume:alpha}, $\zeta$ determines the range of the noise function, while $\epsilon$ is an upper bound to the variance of the noise function. For any function the variance will be smaller or equal to the range. Hence, the condition $4\zeta^4 > B^2 \epsilon$ is reasonably easy to satisfy and we may obtain a quantum advantage for a broad concept class of functions. 
In Appendix \ref{appComb}, we combine the algorithms for the kernel matrix estimation and the inner loop estimations, both for the classical and quantum cases.

\section{Applications for learning two-layer neural networks}
\label{sec:application}

In this section, we describe how to use algorithms described above to learn neural networks with two nonlinear layers in the $p$-concept model.
Following previous works \cite{pmlr-v99-goel19b,10.5555/3305890.3305989}, first consider an one-layer neural network with $k$ units
\begin{align}
\mathcal{N}_1:\mathbf{x} \rightarrow \sum_{i=1}^k \mathbf{b}_i\sigma(\mathbf{a}_i\cdot\mathbf{x}),
\end{align}
where $\mathbf{x}\in\mathbb{R}^n$, $\mathbf{b}\in \mathbb{B}^{k}$, $\mathbf{a}_i\in\mathbb{B}^{n}$ for $i\in\{1,\dots,k\}$, and $\sigma:\mathbb{R}\rightarrow\mathbb{R}$ is the activation function.
In the following, we consider the sigmoid function $\sigma=1/(1+e^{-x})$ as the activation function.
Subsequently, we define a neural work with two nonlinear layers with one unit in the second layer
\begin{align}
\mathcal{N}_2:\mathbf{x} \rightarrow \sigma'\left(\sum_{i=1}^k \mathbf{b}_i\sigma(\mathbf{a}_i\cdot\mathbf{x})\right),
\end{align}
where $\sigma':\mathbb{R}\rightarrow\mathbb{R}$ is an $L$-Lipschitz non-decreasing function.
Ref.~\cite{pmlr-v99-goel19b} proved the following lemma, which states that such kinds of neural networks are efficiently $p$-concept learnable.

\begin{lemma}[Efficient p-concept learning for two-layer neural networks \cite{pmlr-v99-goel19b}]
\label{lemma_nn}
Consider samples $\left(\boldsymbol{x}_i, y_i\right)_{i=1}^m$ drawn i.i.d.~from a distribution $\mathcal{D}$ on $\mathbb{B}^{n} \times[0,1]$ such that $E[y \mid \boldsymbol{x}]=\mathcal{N}_2(\boldsymbol{x})$ with $\sigma^{\prime}: \mathbb{R} \rightarrow[0,1]$ is a known L-Lipschitz non-decreasing function and $\sigma$ is the sigmoid function. There exists an algorithm that outputs a hypothesis $h$ such that with probability $1-\delta$,
$$
\mathbb{E}_{\boldsymbol{x}, y \sim \mathcal{D}}\left[\left(h(\boldsymbol{x})-\mathcal{N}_2(\boldsymbol{x})\right)^2\right] \leq \epsilon,
$$
for $m=\left(\frac{k L}{\epsilon}\right)^{O(1)} \cdot \log^2 (1 / \delta)$. The algorithm runs in time polynomial with $m$ and $n$.
\end{lemma}
The statement of the previous work mentions the training sample complexity $m_1$ instead of total sample complexity $m$, where the latter depends on $\log^2 (1 / \delta)$.
Here, we focus on the total sample complexity.

We assume corresponding query access to the kernel matrices, both for the classical and the quantum scenario, respectively.
We focus on the comparison of time complexity for \textsc{Alphatron\_with\_Kernel} and \textsc{Quantum\_Alphatron}.
The sigmoid function can be uniformly approximated by low-degree polynomials.
By Lemma $8$ and $12$ in Ref.~\cite{pmlr-v99-goel19b}, we have that there exists a constant $C_{\rm sig}$ and a $\mathbf{v} \in \mathcal{H}_d$ with $\|\mathbf{v}\|_2\leq (\sqrt{k}/\epsilon_0)^{C_{sig}}$ for $d=O(\log (1 / \epsilon_0))$ such that
\be \label{eqApproxLayer1}
\left|\sum_{i=1}^k \mathbf{b}_i\sigma(\mathbf{a}_i\cdot\mathbf{x})-\langle \mathbf{v},\psi(\mathbf{x})\rangle \right|\leq \sqrt{k}\epsilon_0,
\ee
for every $\mathbf{x}\in \mathbb{B}^n$.
Hence, in this case we have
$\mathcal{N}_1(\mathbf{x})=\langle\mathbf{v}, \psi(\mathbf{x})\rangle+\xi(\mathbf{x})$ for some function $\xi: \mathcal{X} \rightarrow[-\sqrt{k}\epsilon_0, \sqrt{k}\epsilon_0]$,
and $\mathbb E[y\mid\mathbf{x}]=\sigma'(\langle\mathbf{v}, \psi(\mathbf{x})\rangle+\xi(\mathbf{x}))$ with $\mathbb{E}[\xi(\mathbf{x})^2]\leq k\epsilon_0^2$. 
From these, we have $\zeta=\sqrt{k}\epsilon_0$, $B=(\sqrt{k}/\epsilon_0)^{C_{\mathrm{sig}}}$ and $\epsilon=k\epsilon^2_0$.

Based on the discussion in Section \ref{sec:discuss}, we can compare the value of $4\zeta^4$ and $B^2\epsilon$ to determine the possibility for a quantum speedup.
In this case, we have $4\zeta^4 = 4 (\sqrt{k}\epsilon_0)^4$ and $B^2 \epsilon=k^2(\sqrt{k}/\epsilon_0)^{2C_{\rm sig}-2}$.
The condition for the quantum speedup is $4\zeta^4>B^2\epsilon$, which is equivalent to $4>k^2(\sqrt{k}/\epsilon_0)^{2C_{\rm sig}-2}$.
As $k\geq 1$, $0<\epsilon_0<1$, and $C_{\rm sig}>2$ \cite{10.5555/2968826.2968922}, in general the inequality will not hold, and hence we will not achieve a quantum speedup for learning two-layer neural networks with the sigmoid function.

In the following final part, we discuss a type of neural network which is in a regime where there is a quantum advantage according to Section \ref{sec:discuss}.
The key idea is that if the neural network has some intrinsic errors, these errors dominate and we may achieve a quantum speedup. We first define the neural-network architecture.

\begin{defn}[Erroneous neural networks]
Let $\gamma_1>0$ and $\gamma_2>0$ and let there be given a function $\xi_1(\mathbf{x})\in [-\gamma_1,\gamma_1]$ and $\mathbb{E}[\xi_1^2(x)]\leq \gamma_2$.
We say that we have an $(\gamma_1,\gamma_2)$-erroneous one-layer neural network with $k$ units if it can be written in the form
\begin{align}
\mathcal{N}^{\rm err}_{1}:x\rightarrow \sum_{i=1}^k \mathbf{b}_i\sigma(\mathbf{a}_i\cdot\mathbf{x})+\xi_1(\mathbf{x}),
\end{align}
where $\mathbf{x}\in\mathbb{R}^n$, $\mathbf{b}\in \mathbb{B}^{k}$, $\mathbf{a}_i\in\mathbb{B}^{n}$ for $i\in\{1,\dots,k\}$, and $\sigma:\mathbb{R}\rightarrow\mathbb{R}$ is the sigmoid function.
Based on the one-layer network, we define a two-layer $(\gamma_1,\gamma_2)$-erroneous neural network with
\begin{align}
\mathcal{N}^{\rm err}_{2}:x\rightarrow \sigma'\left(\sum_{i=1}^k \mathbf{b}_i\sigma(\mathbf{a}_i\cdot\mathbf{x})+\xi_1(\mathbf{x})\right),
\end{align}
where $\sigma':\mathbb{R}\rightarrow\mathbb{R}$ is an $L$-Lipschitz non-decreasing function.
\end{defn}
This neural network allows for a quantum advantage for learning it. 
\begin{theorem}
There is a quantum speedup for learning two-layer $(\gamma_1,\gamma_2)$-erroneous neural networks based on the $p$-concept learning model if $\gamma_1> \max \{\frac{1}{\sqrt{2}}(\sqrt{k}/\epsilon_0)^{C_{\rm sig}}, \frac{1}{\sqrt{2}}\sqrt{\gamma_2}\}$.
\end{theorem}

\begin{proof}
Using Eq.~\ref{eqApproxLayer1}, we have that $\forall x\in \mathcal{X}$, $\mathcal{N}^{\rm err}_1=\langle \mathbf v,\psi(\mathbf x) \rangle+\xi'(\mathbf x)$ for some function $\xi':\mathcal{X}\rightarrow 
[-\zeta_{\rm err},\zeta_{\rm err}]$ with $\zeta_{\rm err}:=\sqrt{k}\epsilon_0+\gamma_1$ and with $\mathbb{E}[\xi'^2]\leq (\gamma_2+k\epsilon_0^2+2\gamma_1\sqrt{k}\epsilon_0) =: \epsilon_{\rm err}$.
The $\ell_2$-norm of the hypothesis vector $\|\boldsymbol{v}\|_2$ is bounded by $B_{\rm err}:=(\sqrt{k}/\epsilon_0)^{C_{\rm sig}}$.
To achieve the quantum speedup, we need to satisfy $4\zeta_{\rm err}^4> B_{\rm err}^2\epsilon_{\rm err}$.
It suffices to satisfy
\begin{align}
4 (\sqrt{k}\epsilon_0+\gamma_1)^4 > (\sqrt{k}/\epsilon_0)^{2C_{\rm sig}}(\gamma_2+k\epsilon_0^2+2\gamma_1\sqrt{k}\epsilon_0).
\end{align}
This inequality holds when the following holds
\begin{align*}
2(\sqrt{k}\epsilon_0+\gamma_1)^2>& (\sqrt{k}/\epsilon_0)^{2C_{\rm sig}},\\
2(\sqrt{k}\epsilon_0+\gamma_1)^2>& (\gamma_2+k\epsilon_0^2+2\gamma_1\sqrt{k}\epsilon_0).
\end{align*}
These two inequalities can be achieved if $\gamma_1> \max \{\frac{1}{\sqrt{2}}(\sqrt{k}/\epsilon_0)^{C_{\rm sig}}, \frac{1}{\sqrt{2}}\sqrt{\gamma_2}\}$.
\end{proof}
Hence, we have shown that the original two-layer neural network does not admit a quantum speedup with our approach, while a noisy version of the two-layer network does. 

\section{Acknowledgements}

This work was supported by the Singapore National Research Foundation, the Prime Minister’s Office, Singapore, the Ministry of Education, Singapore under the Research Centres of Excellence programme under research grant R 710-000-012-135.
In addition, this research is supported by the National Research Foundation, Singapore under its CQT Bridging Grant.

\appendix

\section{Lipschitz condition for multinomial kernel function}
\label{appLip}

\begin{lemma} \label{lemmaLipschitz}
Let  $f : \mathbb R \to \mathbb R$ be defined by $f(z) = \frac{1}{d+1}\sum_{i =0}^d z^i $. Then $f(z)$ is Lipschitz continuous with Lipschitz constant $L \in \Ord{d z_0^d + d z_0 + \frac{1}{d}}$, that is
\be
\vert f(z) - f(z')\vert \leq L \vert z - z'\vert,
\ee
for all $z, z' \in [-z_0, z_0]$.
\end{lemma}
\begin{proof}
First, $\frac {d}{dz} f(z) = \frac{1}{d+1}\sum_{i =0}^{d-1} (i+1) z^i \le \frac{1}{d+1} (1 +  \sum_{i =1}^{d-1} d z^i)$.
When $0 < z < 1$, $z^i \le z$ for $1 \le i \le d-1$. And when $z \ge 1$, $z^i \le z^d$ for $1 \le i \le d-1$.
Thus $z^i \le z + z^d$ for $1 \le i \le d-1$.
Hence $L \leq \max_{z\in[-z_0,z_0]} |\frac {d}{dz} f(z)| \le \frac{d}{dz} f(z) \mid_{z = z_0} \in \Ord{d z_0^d + d z_0 + \frac{1}{d}}$.
\end{proof}
Note that throughout this paper, it always holds that $z_0 = 1$. In this case,
the Lipschitz constant is bounded by $\Ord{d}$.

\section{Proof of Theorem~\ref{thm:alpha} \label{proof.alphatron_approx}}

We first introduce several definitions and lemmas for proving the theorem. 
Given the coefficients $\mathbf{\alpha}_i$,
we generate a hypothesis vector $\bm{v}(\mathbf{\alpha})$
in the feature space by taking the linear combination over vectors $\psi_d (\bm{x}_i)$.

\begin{defn}[Auxiliary definitions] \label{def:hypovec}
In the setting of Definitions \ref{assume:alpha} and \ref{assume:param},
define the generated hypothesis mapping $\bm v: \mathbb R^{m_1} \to \mathbb R^{n_d}$ as
the linear combination
\be 
\bm {v}(\mathbf{\alpha}) := \sum_{i=1}^{m_1} \alpha_i \psi_d (\bm{x}_i).
\ee
In addition, define $\beta \in \mathbb{R}^{m_1}$ as
\be
\beta_i := \frac{1}{m_1} (y_i - u(\langle \textbf{v}, \psi_d (\bm{x}_i)\rangle ) + \xi(\bm{x}_i)),
\ee
using $\textbf v$ from the concept class
and
$\Delta := \bm{v} (\beta)$ with the norm
$\eta := \Vert \Delta \Vert_2$.
Finally, define $\rho := \frac{1}{m_1} \sum_{i=1}^{m_1} \xi (\bm{x}_i)^2$ as the average quadratic noise over the input data.
\end{defn}
We note the subtle difference of the symbols $\bm v(\alpha)$ and $\Hyp{v}$ but 
emphasize that $\bm v(\alpha)$ will always have the parenthesis with the input value, while $\Hyp{v}$ is static and fixed by the element of the concept class. 
To adapt to matrices $K_{ij}$ (with dimension $m_1 \times m_1$) and $K'_{ij}$ (with dimension $m_2 \times m_1$),
Definition~\ref{def:hypofunc} and Lemma~\ref{lem:hypomatrixerror} are stated in general terms.

\begin{defn} [Hypothesis function] \label{def:hypofunc}
Let $u:\mathbb{R} \to [0, 1]$ be an 
$L$-Lipschitz function.
Define the hypothesis function $g: \mathbb R^{N_2} \times \mathbb R^{N_1 \times N_2} \times [N_1] \to [0,1]$ as
\be
{g}(\mathbf{\alpha}, M, i) := u\left (\sum_{j=1}^{N_2} {\alpha}_j M_{ij}\right ).
\ee
\end{defn}
In Lemma~\ref{lem:hypomatrixerror}, we show that if we have a good enough estimation $\widetilde M$ for the matrix $M$,
then the estimated result $g(\mathbf{\alpha}, \widetilde M, i)$ is not too far from the exact value $g(\mathbf{\alpha}, M, i)$, with a dependence on $\Vert \alpha \Vert_1$. 

\begin{lemma}\label{lem:hypomatrixerror}
Let $M, \tilde M \in \mathbb R^{N_1 \times N_2}$ be matrices of dimension $N_1 \times N_2$.
Let $\epsilon \in (0,1)$.
Let $u: \mathbb{R} \to [0, 1]$ be $L$-Lipschitz.
If $\max_{i \in[N_1], j\in[N_2]}\vert M_{ij} - \tilde M_{ij}\vert \leq \epsilon$,
then for all $\alpha \in \mathbb R^{N_2}$,
we have 
\[ \max_{i \in [N_1]} \left \vert g(\alpha, M,i) - g(\alpha, \tilde M, i)\right \vert \le  L \epsilon\Vert \alpha \Vert_1. \]
\end{lemma}
\begin{proof}
For all $i \in [N_1]$,
\be
\vert g(\alpha, M,i) - g(\alpha, \tilde M, i)\vert &\leq& L \left \vert \sum_{j=1}^{N_2} {\alpha}_j (M_{ij}- \tilde M_{ij} ) \right \vert \\
& & (\textrm{by the $L$-Lipschitz conditon of $u$}) \nonumber\\
&\leq& L \epsilon \sum_{j=1}^{N_2} \left \vert {\alpha}_j \right \vert \\
& & (\textrm{by assumption}) \nonumber\\
& = &  L \epsilon \Vert \alpha \Vert_1.
\ee
\end{proof}

In the following lemma we show that if we update $\alpha^t$ according to 
the \textsc{Alphatron} algorithm, then
the max norm of $\alpha^t$ can be bounded in terms of $T$, $\lambda$ and $m_1$.
\begin{lemma}\label{lem:alphabound}
For arbitrary $K \in \mathbb R^{m_1\times m_1}$,  $y\in[0,1]^{m_1}$, and $\lambda \in \mathbb R_+$, if the initial vector is $\alpha^0 = \bm 0$,
then by performing the updates $\alpha^{t+1}_i \gets \alpha^{t}_i + \frac{\lambda}{m_1} (y_i - g( \alpha^t, K,i))$, for each entry
$T$ times, we have $\max_i \vert \alpha^T_i\vert \le \frac{T\lambda}{m_1}$.
\end{lemma}
\begin{proof}
We prove the statement by induction. The base case is obviously true.
Note that $y$ is from $[0,1]$ and the range of $g$ is also $[0,1]$.
Hence, 
\[ \left \vert {\alpha}^{t}_i \right\vert \le \left \vert{\alpha}^{t-1}_i \right \vert + \frac{\lambda}{m_1}\left \vert y_i - {g}(\alpha^{t-1},K,i)\right\vert
\le \frac{(t-1)\lambda}{m_1} + \frac{\lambda}{m_1} = \frac{t\lambda}{m_1}.\]
\end{proof}

We state the convergence result from the original work Ref.~\cite{pmlr-v99-goel19b}
in Lemma~\ref{lem:alphatron}
before proving our own Lemma~\ref{lemConvergenceModified}.
Intuitively, these lemmas prove that we indeed make progress towards the target by each iteration.
The norm $\Vert {\bm{v}}(\omega) - \textbf{v}\Vert^2_2$ measures the distance between the current vector and the target.
If we show that this value decreases as we run the algorithm and if we lower bound this value in term of the empirical error $\hat{\varepsilon}$ of the current vector,
then either we made progress, or the quality of the current vector is already good enough.

\begin{lemma}[Convergence of Algorithm~\ref{algAlphatron} from \cite{pmlr-v99-goel19b}]\label{lem:alphatron}
Consider Definition \ref{assume:alpha} and \ref{assume:param} for the setting and the algorithm parameters,
as well as the training labels $y \in [0,1]^{m_1}$,
and the kernel matrix $K \in [-1,1]^{m_1\times m_1}$.
For any vector $\omega \in \mathbb{R}^{m_1}$,
let $\omega' \in \mathbb{R}^{m_1}$ be the vector defined as
\[ \omega'_i := \omega_i + \frac{\lambda}{m_1} (y_i - g(\omega, K,i)). \]
Let $h$ be the hypothesis function defined as
$h(\bm{x}) = u(\langle \bm{v}(\omega), \psi_d(\bm{x}) \rangle)$,
and recall from Definition~\ref{def:hypovec} that $\eta = \Vert \Delta \Vert_2$.
If $\Vert {\bm {v}}(\omega) - \textbf{v}\Vert_2 \leq B$, for $B > 1$,
and $\eta < 1$, then
\[
\Vert {\bm{v}}(\omega) - \textbf{v}\Vert^2_2 - \Vert \bm v(\omega') - \textbf{v}\Vert^2_2 \geq \frac{1}{L^2}\hat{\varepsilon}(h) - A_{1},
\]
where by definition $A_1 :=  \frac{2}{L}\sqrt{\rho} + \frac{2B\eta}{L} + \frac{\eta^2}{L^2}  + \frac{2\eta}{L^2} $.
\end{lemma}

We claim the following modified convergence result for our Algorithm \ref{algAlphatronApprox}. The difference to the previous lemma is the appearance of a term $-\epsilon_I^2$ in the convergence bound. 
\begin{lemma} \label{lemConvergenceModified}
Consider  Definition \ref{assume:alpha} and \ref{assume:param} for the setting and the algorithm parameters,
as well as the training labels $y \in [0,1]^{m_1}$,
and the kernel matrix $K \in [-1,1]^{m_1\times m_1}$.
Let $\epsilon_I > 0$.
Suppose that $\Gamma_i$ is an estimation of $g(\alpha^t, K,i)$, for all $i \in [m_1]$,
such that
\[\max_{i \in [m_1]} \left \vert g(\alpha^t, K,i) - \Gamma_i \right\vert \le  L \epsilon_I.\]
For any vector $\omega \in \mathbb{R}^{m_1}$,
let $\omega' \in \mathbb{R}^{m_1}$ be the vector defined as
\[ \omega'_i := \omega_i + \frac{\lambda}{m_1} (y_i - g(\omega, K,i)), \]
and $\tilde{\omega} \in \mathbb{R}^{m_1}$ be the vector defined as
\[ \tilde{\omega}_i := {\omega}_i + \frac{\lambda}{m_1} (y_i - \Gamma_i).\]
Let $h$ be the hypothesis function defined as $h(\bm{x}) = u(\langle \bm{v}(\omega), \psi_d(\bm{x}) \rangle)$.
Let $A_1$ be defined as in Lemma~\ref{lem:alphatron}.
Recall that $\eta = \Vert \Delta \Vert_2$.
If $\Vert {\bm {v}}(\omega) - \textbf{v}\Vert_2 \leq B$, for $B > 1$,
and $\eta < 1$, then
\be
\Vert {\bm{v}}(\omega) - \textbf{v}\Vert^2_2 - \Vert \bm v(\tilde \omega) - \textbf{v}\Vert_2^2 \geq
\frac{1}{L^2}\hat{\varepsilon}(h)- A_{1} - \epsilon_I^2. 
\ee
\end{lemma}
\begin{proof}
Recall from Definition~\ref{def:hypovec} that $\bm{v}(\alpha) = \sum_{i=1}^{m_1} \alpha_i \psi(\bm{x}_i)$.
Then we have
\be
\Vert \bm v(\tilde \omega) - \bm{v}( \omega')\Vert_2^2 &=&
\frac{\lambda^2}{m_1^2}  \left \Vert \sum_{j=1}^{m_1} ( \Gamma_j - g(\omega,K,j) )\psi(\bm{x}_j)\right \Vert_2^2 \\
& & \textrm{(expand the definition of \textbf{v}, } \omega'\textrm{, and } \tilde{\omega}\textrm{)} \nonumber \\
&\leq& \frac{\lambda^2}{m_1^2}  \left ( \sum_{j=1}^{m_1} \left \vert \Gamma_j - g(\omega,K,j) \right \vert \Vert \psi(\bm{x}_j)\Vert_2 \right)^2 \\
& & \textrm{(by triangle's inequality)} \nonumber \\
&\leq& \frac{\lambda^2 L^2 \epsilon_I^2 }{m_1^2} m_1^2  \\
& & \textrm{(by assumptions)} \nonumber \\
&=& \epsilon_I^2 .
\ee
Therefore, using the triangle inequality, we can deduce that
\be
\Vert \bm{v}(\tilde {\omega}) - \Hyp{v}\Vert_2^2 &\leq& \Vert \bm{v}(\tilde{\omega}) - \bm{v}( \omega')\Vert_2^2 +
\Vert \bm {v}(\omega') - \Hyp{v}\Vert_2^2 \\
&\leq& \epsilon_I^2 + \Vert \bm{v}(\omega') - \Hyp{v}\Vert_2^2. \label{eq:lemConvMod}
\ee
Note that except for the vector $\tilde{\omega}$ and the related conditions, Lemma~\ref{lemConvergenceModified} has
the same settings as Lemma~\ref{lem:alphatron}.
Thus by the conclusion of Lemma~\ref{lem:alphatron},
we can lower bound $\Vert {\bm{v}}(\omega) - \Hyp{v}\Vert^2_2 - \Vert \bm{v}(\omega') - \Hyp{v}\Vert_2^2$ by
$\frac{1}{L^2} \hat{\varepsilon}(h) - A_1$.
Together with equation~(\ref{eq:lemConvMod}), we obtain the required bound.
\end{proof}

In the last step of Algorithm \ref{algAlphatronMain},
we pick the hypothesis with
the minimum $\hat{\err}$ value. We also lose accuracy here as we use entries from the approximation of the matrix $K'$.
In Lemma~\ref{lemma:ValidationError}, we show that these errors are acceptable.
\begin{lemma}
\label{lemma:ValidationError}
Consider the training data samples $(\bm{x}_i, y_i) \in \mathbb{B}^{n} \times [0, 1]$, for $i \in [m_1]$,
validation data samples $(\bm{a}_i, b_i) \in \mathbb{B}^{n} \times [0, 1]$, for $i \in [m_2]$,
the corresponding kernel matrix $K' \in [-1,1]^{m_2 \times m_1}$ where
$K'_{ij} = \mathcal{K}_d(\bm{a}_i, \bm{x}_j)$, and the vectors $\alpha^t \in \mathbb R^{m_1}$, for $t \in [T]$.
Let $\epsilon_I > 0$.
Suppose that $\Gamma^t_i$ is an
estimation of $g(\alpha^t, K',i)$, for all $t \in [T]$ and $i \in [m_2]$,
such that
\[ \max_{i \in [m_2]} \left \vert g(\alpha^t, K',i) - \Gamma^t_i \right\vert \le  L \epsilon_I.\]
We define the hypothesis functions $h^t$ as $h^t (\bm{x}) = u(\langle \bm{v}(\alpha^t), \psi_d(\bm{x}) \rangle)$.
Let $$\tilde t = \arg\min_{t \in [T]} \left\{ \frac{1}{m_2} \sum_{i=1}^{m_2} (\Gamma^{t}_i - b_i)^2 \right\}$$
and let $t' = \arg\min_{t \in [T]} \hat{\err}(h^t)$.
Then
\be
\hat {\err}(h^{\tilde t}) - \hat{\err}(h^{t'})  \in \Ord{L\epsilon_I}.
\ee
\end{lemma}
\begin{proof}
By Definition~\ref{def:hypofunc} and the assumption on the kernel matrix $K'$,
we obtain
\be
h^t(\bm{a}_i) = u(\langle \bm{v}(\alpha^t), \psi_d(\bm{a}_i) \rangle) = g(\alpha^t, K', i). \label{eq:lemValidError}
\ee
We have $m_2$ samples for validation.
Recall the definition of the empirical error \\ $\hat{\err}(h) = \frac{1}{m_2}\sum_{i=1}^{m_2} (h(\bm{a}_i) - b_i)^2$ in the Preliminary.
For fixed $t$, 
\be
\left \vert \hat{\err}(h^t) - \frac{1}{m_2} \sum_{i=1}^{m_2} (\Gamma^t_i - b_i)^2 \right\vert & = & \frac{1}{m_2}  \left \vert \sum_{i=1}^{m_2} (g(\alpha^t, K', i)^2 - (\Gamma^t_i)^2 - 2b_i (g(\alpha^t, K', i) - \Gamma^t_i ) ) \right \vert\\
& & \textrm{(by empirical error with respect to $(\bm{a}_i, b_i)$ and equation~(\ref{eq:lemValidError}))} \nonumber \\
& \le & \max_{i \in [m_2]} \left \vert g(\alpha^t, K', i)^2 - (\Gamma^t_i)^2 - 2b_i (g(\alpha^t, K', i) - \Gamma^t_i ) \right \vert\\
& & \textrm{(the maximum is at least the average)} \nonumber \\
& \le & \max_{i \in [m_2]} \left \vert (g(\alpha^t, K', i) + \Gamma^t_i - 2b_i )(g(\alpha^t, K', i) - \Gamma^t_i )\right \vert\\
& \le & \max_{i \in [m_2]} \left \vert g(\alpha^t, K', i) + \Gamma^t_i - 2b_i \right \vert \cdot \vert g(\alpha^t, K', i) - \Gamma^t_i \vert\\
& \le & \max_{i \in [m_2]} 2\left \vert g(\alpha^t, K', i) - \Gamma^t_i \right\vert\\
& & \left ( g(\alpha^t, K', i), \Gamma^t_i, b_i \textrm{ are in range } [0, 1] \right ) \nonumber \\
& \le & 2 L \epsilon_I. \label{eq:lemmaVE}\\
& & \textrm{(by assumption)} \nonumber
\ee
Hence, we have both
\be
\hat{\err}(h^{t'}) + 2L\epsilon_I \ge \frac{1}{m_2} \sum_{i=1}^{m_2} (\Gamma^{t'}_i - b_i)^2, \label{eq:lemValidError1}
\ee
when $t = t'$ in (\ref{eq:lemmaVE}), and
\be
\frac{1}{m_2} \sum_{i=1}^{m_2} (\Gamma^{\tilde t}_i - b_i)^2 \ge \hat{\err} (h^{\tilde t}) - 2L\epsilon_I,\label{eq:lemValidError2}
\ee
when $t = \tilde{t}$ in (\ref{eq:lemmaVE}). And by the minimization of $\tilde t$,
\be
\frac{1}{m_2} \sum_{i=1}^{m_2} (\Gamma^{t'}_i - b_i)^2 \ge \frac{1}{m_2} \sum_{i=1}^{m_2} (\Gamma^{\tilde t}_i - b_i)^2.\label{eq:lemValidError3}
\ee
From the last three equations by transitivity 
we deduce the required $\hat {\err}(h^{\tilde t}) - \hat{\err}(h^{t'})\in \Ord{L\epsilon_I}$.
\end{proof}

Now, we are going to prove the main theorem.
As the original work, the theorem requires a generalization bound which involves the Rademacher complexity of the function class considered here. The required result is Theorem~\ref{thm:Rademacher} in the Appendix \ref{appRade}.

\begin{proof}[Proof of Theorem \ref{thm:alpha}]
We first consider the success probability of the approximation part.
According to Algorithm~\ref{algAlphatronApprox},
we estimate each inner product with success probability $1-\delta_K/(m_1^2 + m_1m_2)$.
A union bound for these $m_1^2 + m_1 m_2$ estimations gives the total success probability to be at least $1-\delta_K = 1 - \delta$.
Then it is sufficient to show that the main body itself succeeds with probability at least $1 - \delta$ and indeed produces a good enough hypothesis.

The remaining proof consists of three parts. In the first part, we show that there exists $t^*  \in [T]$ such that the empirical error of $h^{t^*}$ is good enough.
In the second part, we show using the Rademacher complexity that for any specific hypothesis in the concept class we introduced, the generalization error is not very far from empirical error. 
In the third part, we show that by using $m_2$ additional samples to validate all generated hypotheses, as done in the algorithm, we are able to find a hypothesis with similar error as $h^{t^*}$.

First, we show that there exists a good enough hypothesis according to the empirical error.
Recall the notations $\Delta$ and $\rho$ in Definition~\ref{def:hypovec}, and
let $\eta =\Vert \Delta \Vert_2$.
Ref.~\cite{pmlr-v99-goel19b} shows that $\eta \le \frac{1}{\sqrt{m_1}} (1 + \sqrt{2 \log (1/\delta)})$, and $\rho \le \sqrt{\epsilon} + \Ord{\sqrt[4]{\frac{\log (1/\delta)}{m_1}}}$ by Hoeffding's inequality.
Since $\eta$ and $\rho$ only depend on the setting in Definition \ref{assume:alpha}, the modification done in Algorithm~\ref{algAlphatronApprox} compared to Algorithm~\ref{algAlphatronPre} keeps the bounds for $\eta$ and $\rho$ the same. 
Hence, we use the same bounds in this proof.

In the Algorithm~\ref{algAlphatronMain}, which is used in Algorithm~\ref{algAlphatronApprox},
assume we are presently at the iteration $t$ for computing the vector $\widetilde \alpha^{t+1}$ from $\widetilde \alpha^t$.
In this proof, we use the tilde above the $\alpha$ to emphasize that we indeed construct a different sequence $(\widetilde \alpha^t)_{t\in [T]}$  compared to the sequence $(\alpha^t)_{t\in [T]}$ of  Algorithm~\ref{algAlphatronMain} with the exact kernel matrices.
One of the following two cases is satisfied,
\be
\textbf{ Case 1} : \left\Vert\textbf{v}\left (\widetilde \alpha\super{t}\right) - \Hyp{v} \right\Vert_2^2 - \left\Vert\textbf{v}\left(\widetilde\alpha\super{t+1}\right)- \Hyp{v}\right\Vert_2^2 &>& \frac{ B \eta}{L}, \label{eqConvergenceCase1}\\
\textbf{ Case 2}: \left\Vert \textbf{v}\left(\widetilde \alpha\super{t}\right) - \Hyp{v}\right\Vert_2^2 - \left\Vert\textbf{v}\left(\widetilde\alpha\super{t+1}\right) - \Hyp{v}\right\Vert_2^2 &\leq& \frac{ B \eta}{L}.
\label{eqConvergenceCase2}
\ee
Let $t^*$ be the first iteration where Case 2 holds. We show that such an iteration exists.
Assume the contradictory, that is, Case 2 fails for each iteration.
Since $\Vert \textbf{v}(\widetilde \alpha\super {0}) - \Hyp{v}\Vert^2_2 = \Vert \mathbf{0} - \Hyp{v}\Vert^2_2 \leq B^2$ by assumption, however,
\be
B^2 &\geq& \left\Vert\textbf{v}\left(\widetilde \alpha\super{0}\right) - \Hyp{v}\right\Vert_2^2  \geq \left \Vert\textbf{v}\left(\widetilde \alpha\super{0}\right) - \Hyp{v}\right\Vert_2^2 - \left \Vert\textbf{v}\left(\widetilde \alpha\super{k}\right) - \Hyp{v}\right \Vert_2^2 \\
&=& \sum_{t=0}^{k-1} \left (\left\Vert\textbf{v}\left(\widetilde \alpha\super{t}\right) - \Hyp{v}\right\Vert_2^2 - \left\Vert\textbf{v}\left(\widetilde\alpha\super{t+1}\right)- \Hyp{v}\right\Vert_2^2 \right) \geq \frac{ kB\eta}{L},
\ee
for $k$ iterations.
Hence, in at most $\frac{BL}{\eta}$ iterations Case 1 will be violated and Case 2 will have to be true.
By Assumption~\ref{assume:param} and the bound on $\eta$, we have that
$T \ge \frac{BL}{\eta}$, and then $t^* \in [T]$ must exist such that Case 2 is true.
For all $t \in [T]$, define the hypothesis function $h^t$ as $h^t(\bm{x}) = u(\langle \bm{v}(\widetilde{\alpha}^t), \psi_d(\bm{x}) \rangle)$.
By Theorem~\ref{thmApprox}, we have that $\max_{ij} \left \vert K_{ij} - \widetilde{K}_{ij} \right \vert \le \epsilon_K$.
Define the shorthand $\Gamma_i := g(\widetilde{\alpha}^{t^*}, \widetilde K, i)$, with the hypothesis function from Definition \ref{def:hypofunc}.
Then, with Lemma~\ref{lem:hypomatrixerror},
we obtain $\max_{i \in [m_1]} \left \vert g(\widetilde{\alpha}^{t^*}, K,i) - \Gamma_i\right \vert \le  L \epsilon_K \left\Vert \widetilde{\alpha}^{t^*} \right\Vert_1$.
Then, by Lemma~\ref{lemConvergenceModified} with $\omega = \widetilde{\alpha}^{t^*}$ and $\epsilon_I = \epsilon_K \Vert \widetilde{\alpha}^{t^*} \Vert_1$,
we obtain
\be
\left \Vert \textbf{v}\left (\widetilde \alpha\super{t^*}\right ) - \Hyp{v}\right \Vert_2^2 - \left\Vert\textbf{v}\left(\widetilde\alpha\super{t^*+1}\right) - \Hyp{v}\right\Vert_2^2
\geq \frac{1}{L^2}\hat{\varepsilon}(h^{t^\ast})- A_{1} - \epsilon_K^2 \left\Vert \widetilde \alpha^{t^\ast} \right\Vert_1^2.
\ee
Note that Case 2 holds for the iteration $t^*$.
Together with the upper bound in Eq.~(\ref{eqConvergenceCase2}),   
it holds by transitivity that
\be
\frac{ B \eta }{L} \geq \frac{1}{L^2} \hat{\varepsilon}\left (h^{t^\ast}\right)- A_{1} - \epsilon_K^2 \left\Vert \widetilde \alpha^{t^\ast} \right\Vert_1^2,
\ee
which implies that 
\be
\hat \varepsilon\left(h^{t^\ast}\right) \leq B L \eta + L^2 A_{1} + L^2 \epsilon_K^2 \left\Vert \widetilde \alpha\super{t^\ast} \right\Vert_1^2.
\ee
Recall the definition of $A_2 \triangleq L\sqrt{\epsilon} + L \zeta \sqrt[4]{\frac{\log(1/\delta)}{m_1}} + BL\sqrt{\frac{\log(1/\delta)}{m_1}}$.
Using the known bounds for $\eta$ and $\rho$, we have 
\be
\hat{\varepsilon}\left (h^{t^*}\right) \in \Ord{A_2+
L^2 \epsilon_K^2 \left \Vert \widetilde \alpha\super{t^\ast}\right \Vert_1^2 }.
\ee
The last term can be bounded  
as
\be
L^2 \epsilon_K^2 \left \Vert \widetilde \alpha\super{t^\ast} \right\Vert_1^2 \leq T^2 \epsilon_K^2 , 
\ee
where we use $\Vert \widetilde \alpha\super{t^\ast} \Vert_1 \leq T/L$ from Lemma \ref{lem:alphabound}.

As a next step, we would like to bound $\varepsilon(h^{t^*})$ in terms of $\hat{\varepsilon}(h^{t^*})$.
An argument based on the Rademacher complexity gives us the same bound as in the original
work \cite{pmlr-v99-goel19b}.
Define a function class
$\mathcal{Z} = \{ \textbf{x} \rightarrow u(\langle \bm{z}, \psi_d(\bm{x})\rangle) - \mathbb{E}_y [y | \bm{x}] : \Vert \bm{z} \Vert_2 \le 2B \}$.
Ref.~\cite{pmlr-v99-goel19b} shows that the Rademacher complexity of $\mathcal{Z}$ is $\mathcal{R}_m (\mathcal{Z}) \in \Ord{BL\sqrt{1/m}}$.

Let us show that $\left \Vert\textbf{v}\left (\widetilde \alpha\super{t^*}\right) \right\Vert_2$ satisfies the norm bound $2B$, same as the $\bm z$ in class $\mathcal Z$. Note that in the first $t^* - 1$ iterations, Case 1 holds. By Eq.~(\ref{eqConvergenceCase1}), we have 
$\left \Vert\textbf{v}\left(\widetilde \alpha\super{t}\right) - \Hyp{v}\right\Vert_2^2 - \left\Vert\textbf{v}\left(\widetilde\alpha\super{t+1}\right)- \Hyp{v}\right\Vert_2^2 > \frac{ B \eta}{L} \ge 0$,
for $t \in [t^* - 1]$.
In other words, the distance between $\textbf{v}\left(\widetilde \alpha\super{t}\right)$ and $\Hyp{v}$ decreases when $t$ increases in $[t^* - 1]$.
Thus, we conclude that
\be
\left \Vert\textbf{v}\left(\widetilde \alpha\super{t^*}\right) - \Hyp{v}\right\Vert_2^2 \le \left\Vert\textbf{v}(\widetilde\alpha\super{0})- \Hyp{v}\right\Vert_2^2 = \left\Vert \Hyp{v} \right\Vert_2^2 \le B^2,
\ee
and hence by the triangle inequality, $\Vert\textbf{v}(\widetilde \alpha\super{t^*}) \Vert_2 \le 2B$.
Denote $f(\bm{x})$ as $h^{t^*}(\bm{x}) - \mathbb{E}_y [y | \bm{x}]$.
Since $h^{t^*}(\bm{x}) = u\left(\left \langle \textbf{v}\left(\widetilde \alpha\super{t^*}\right), 
\psi_d(\bm{x})\right \rangle \right)$,
the function $f(\bm{x})$ 
is an element of $\mathcal{Z}$.
Define the loss function $\mathcal{L} : [0,1] \times [0,1] \to [-1, 1]$ as $\mathcal{L} (a, a') = a^2$ which ignores the second argument.
According to Theorem~\ref{thm:Rademacher} in the Appendix \ref{appRade}, with $ \mathcal{J}(f, \mathcal{D}) =\varepsilon(h^{t^*})$
and $\hat{\mathcal{J}}(f, S) = \hat{\varepsilon}(h^{t^*})$ and $b = 1$,
and with probability $1- \delta$, 
\be
\varepsilon\left (h^{t^*}\right) \leq \hat{\varepsilon}\left(h^{t^*}\right) + \Ord{BL \sqrt{\frac{1}{m_1}}
+  \sqrt{\frac{\log (1/\delta)}{m_1}}} \in \Ord{A_2 + T^2 \epsilon_K^2}.
\ee
The above proof shows the existence of a good hypothesis $h^{t^*}$ for some $t^* \in [T]$.
We define the index of the best hypothesis as
\be t' := \arg\min_{t\in[T]} \varepsilon(h^t), \ee
which immediately implies that
\be
\varepsilon(h^{t'}) \le \varepsilon(h^{t^*}) \in \Ord{A_2 + T^2 \epsilon_K^2}.  \label{eq:conv1}
\ee
In the last part of this proof, we show that at Line~\ref{algLine:valid} of \textsc{Alphatron\_with\_Kernel}, 
we indeed find and output a good enough hypothesis (though this hypothesis may be different from the hypothesis derived from the output of \textsc{Alphatron\_with\_Kernel}).
Our goal is to find a hypothesis $h^t$ which minimizes $\varepsilon(\cdot)$.
However, $\varepsilon(\cdot)$ is hard to compute according to the definition.
From Eq.~(\ref{eq:constErr}), we have that for arbitrary hypotheses $h_1, h_2$,
\be
\varepsilon(h_1) - \varepsilon(h_2) = \err(h_1) - \err(h_2). \label{eq:conv2}
\ee
Hence, we may find the best hypothesis by minimizing $\err(\cdot)$ instead of $\varepsilon(\cdot)$.
Formally, it holds that
\be t' = \arg\min_{t\in[T]} \err(h^t). 
\ee
As we do not know the distribution $\mathcal{D}$, we are unable to compute $\err(\cdot)$.
However, it is possible to compute the empirical version $\hat{\err}(\cdot)$.
In Algorithm \ref{algAlphatronMain}, we use a fresh sample set $(\bm{a}_i, b_i)$ of size $m_2$ as the validation data set,
and we compute the empirical error $\hat{\err}(h^t)$, for each $h^t$, on this data set.
Let $\epsilon' = 1/\sqrt{m_1}$.
For fixed $t$, since $\hat{\err}(h^t)$ is in $[0,1]$, by a Chernoff bound on $m_2 \in \Ord {\log(T/\delta)/(\epsilon')^2}$ samples,
with probability $1 - \delta/T$, we obtain
\be \left \vert \err\left(h^t\right) - \hat{\err}\left(h^t\right)\right \vert \le \epsilon'. \label{eq:errorChernoff} \ee
Then by the union bound, with probability $1 - \delta$, the inequality~(\ref{eq:errorChernoff}) holds simultaneously for all $t \in [T]$. 
Since $\epsilon' \in \Ord{BL\sqrt{\frac{\log 1/\delta}{m_1}}}$, we obtain the bound
\be \left \vert \err\left(h^t\right) - \hat{\err}\left(h^t\right) \right \vert \in \Ord{A_2}, \label{eq:conv3}
\ee
for all $t \in [T]$.	
However, in  Algorithm~\ref{algAlphatronApprox}, to find the hypothesis with the minimum empirical error, we use the estimated kernel matrix $\widetilde{K}'$ instead of the exact inner products. 
Thus, we have additional errors in computing $\hat{\err}(h^t)$.
Let 
\be \widetilde{t} = \arg\min_{t \in [T]} \left\{ \frac{1}{m_2} \sum_{i=1}^{m_2} \left (u\left(\sum_{j=1}^{m_1} \alpha^t_i \cdot \widetilde{K}'_{ij}\right) - b_i\right)^2 \right\}
= \arg\min_{t \in [T]} \left\{ \frac{1}{m_2} \sum_{i=1}^{m_2} \left (g\left (\alpha^t, \widetilde{K}', i\right) - b_i\right )^2 \right\} \ee
be the index of the hypothesis of the output
in Algorithm~\ref{algAlphatronApprox}.
As before, the exact kernel matrix $K'$ is $K'_{ij} = \mathcal{K}_d(\bm{a}_i, \bm{x}_j)$.
According to Theorem \ref{thmApprox},
we have $\left \vert \widetilde{K}'_{ij} - K'_{ij} \right\vert \le \epsilon_K$.
Via Lemma~\ref{lem:hypomatrixerror}, we obtain
$\vert g(\alpha^t, \widetilde{K}', i) - g(\alpha^t, K', i) \vert \le L \epsilon_K \Vert \alpha^t \Vert_1$.
By using
$\epsilon_I = \epsilon_K \Vert \alpha^t \Vert_1$ and $\Gamma^t_i = g(\alpha^t, \widetilde{K}', i)$, the upper bound on the estimation error $\vert \hat{\err}(h^{\widetilde{t}}) - \hat{\err}(h^{t'}) \vert$ is shown to be in Lemma~\ref{lemma:ValidationError}, as
\be \left \vert \hat{\err}\left(h^{\widetilde{t}}\right) - \hat{\err}\left(h^{t'}\right) \right\vert \le L \epsilon_I. \label{eq:conv4} \ee 
We have $\Vert \alpha^t \Vert_1 \le T/L$. Thus $\epsilon_I = \epsilon_K \Vert \alpha^t \Vert_1 \le T \epsilon_K/L$.
From Eq.~(\ref{eq:conv3}) and Eq.~(\ref{eq:conv4}), we obtain
\be 
\left \vert \err\left(h^{\widetilde{t}}\right) - \err\left(h^{t'}\right) \right \vert \in \Ord{L \epsilon_I + A_2} \subseteq \Ord{T \epsilon_K + A_2 }. 
\ee
With Eq.~(\ref{eq:conv2}), we obtain $\left \vert \varepsilon\left(h^{\widetilde{t}}\right) - \varepsilon\left(h^{t'}\right)\right \vert \in \Ord{T \epsilon_K + A_2}$. 
Hence, we have
$\varepsilon\left(h^{\widetilde{t}}\right) \le \varepsilon\left(h^{t'}\right) + \Ord{T \epsilon_K + A_2}$. 
From Eq.~(\ref{eq:conv1}), $\varepsilon(h^{t'})$ is bounded by $\Ord{A_2 + T^2 \epsilon_K^2}$. Thus, 
\be
\varepsilon(h^{\widetilde t}) \in \Ord{A_2 + T^2\epsilon_K^2 + T\epsilon_K}.
\label{eq:conv5}
\ee
The union bound of the probabilistic steps of estimating the full kernel matrix, the Rademacher generalization bound, and the Chernoff bound leads to a total success probability of $1-3\delta$.
\end{proof}

Since by definition $\varepsilon(h) \le 1$, for any hypothesis $h$,
it is reasonable to assume that $A_2 \le 1$ if we want a useful bound.
Then, by setting $\epsilon_K = \frac{A_2}{T}$, we have $\epsilon_K T \le 1$.
Thus, $\Ord{\epsilon_K^2 T^2} \subseteq \Ord{\epsilon_K T}$,
and we can simplify the right hand side of Eq.~(\ref{eq:conv5}) to $\Ord{A_2 + \epsilon_K T}$.
From the runtime analysis in Theorem~\ref{thmApprox} and 
the accuracy analysis in Theorem~\ref{thm:alpha},
we have the following corollary. 
\begin{corollary} \label{corApproxLearning}
In the same setting as Theorem~\ref{thm:alpha}, if $L\sqrt \epsilon \le 1$,
and we set $\epsilon_K = \frac{L\sqrt \epsilon}{T} $, 
then Algorithm~\ref{algAlphatronApprox} with probability
at least $1 - 3\delta$ outputs $\alpha^{t_{\rm out}}$ which describes the hypothesis $h^{t_{\rm out}}(\bm{x}) := u\left (\sum_{i=1}^{m_1} \alpha_i^{t_{\rm out}} \mathcal{K}_d(\mathbf{x}, \mathbf{x}_i)\right)$ such that
\[
\varepsilon\left(h^{t_{\rm out}}\right) \in \Ord{A_2},
\]
with a run time of 
\[ \tOrd{ mn + \frac{  m^2 d^2 T^2}{L^2 \epsilon}\log \frac{1}{\delta} + Tm^2 }.\]
\end{corollary}

\section{Classical sampling}
\label{appClassical}
The next facts discuss the construction of a data structure to sample from a vector and the next lemmas discuss the approximation of an inner product of two vectors by sampling. Both  $\ell_1$ and $\ell_2$ cases are required in this work. 
The $\bm{SQ}$ label can be understood as ``sample query". The arithmetic model allows us to assume infinite-precision storage of the real numbers. 
\begin{fact} [$\ell_1$-sampling \cite{Vos91,Wal74}]\label{factSamplingl1}
Given an $n$-dimensional vector $\bm{u} \in \mathbb R^n$,
there exists a data structure to sample an index $j\in[n]$ with probability
$\vert u_j\vert/\Vert \bm{u}\Vert_1$ which can be constructed in time $\tOrd{n}$. One sample can be obtained in time $\Ord{\log n}$. We call this data structure $\bm{SQ1}(\bm u, n)$.
\end{fact}
\begin{fact} [$\ell_2$-sampling \cite{Vos91,Wal74}]\label{factSamplingl2}
Given an $n$-dimensional vector $\bm{u} \in \mathbb R^n$,
there exists a data structure to sample an index $j\in[n]$ with probability $u_j^2/\Vert \bm{u}\Vert^2_2$
which can be constructed in time $\tOrd{n}$. One sample can be obtained in time $\Ord{\log n}$. 
We call this data structure $\bm{SQ2}(\bm u, n)$.
\end{fact}
Next, we show the estimation of inner products via sampling.
The number of samples scales with $1/\epsilon^2$ classically, in contrast to using quantum 
amplitude estimation which scales with $1/\epsilon$.
Lemma \ref{lemmaSamplingl1} is adapted from \cite{Tang2018} and Lemma \ref{lemmaSamplingl2} is taken directly from \cite{Tang2018}.
\begin{lemma}[Inner product with $\ell_1$-sampling]
\label{lemmaSamplingl1}
Let $\epsilon,\delta \in (0,1)$. Given query access to $\bm{v} \in\mathbb R^n$ and
$\bm{SQ1}(\bm u, n)$ access to $\bm{u} \in\mathbb R^n$,
we can determine $ \bm{u}\cdot \bm{v}$ to additive error $\epsilon$ with success probability at least
$1-\delta$ with $\Ord{\frac{\Vert \bm{u} \Vert_{1}^2 \Vert \bm{v} \Vert_{\max}^2}{\epsilon^2} \log \frac{1}{\delta}}$
queries and samples, and $\tOrd{\frac{\Vert \bm{u} \Vert_{1}^2 \Vert  \bm{v}\Vert_{\max}^2}{\epsilon^2} \log \frac{1}{\delta}}$ time complexity.
\end{lemma}
\begin{proof}
Define a random variable $Z$ with outcome ${\rm sgn}(u_j) \Vert \bm{u}\Vert_1 v_j$ with probability $\vert u_j\vert/\Vert \bm{u}\Vert_1$.
Note that $\mathbb E[Z] = \sum_j {\rm sgn}(u_j) \Vert \bm{u}\Vert_1 v_j \vert u_j\vert /\Vert \bm{u}\Vert_1 = \bm{u} \cdot \bm{v}$.
Also, $\mathbb V[Z] \leq \mathbb E[Z^2] = \sum_j \Vert \bm{u}\Vert_1^2 v_j^2 \vert u_j\vert/\Vert \bm{u}\Vert_1 \leq
\Vert \bm{u}\Vert_1^2 \Vert  \bm{v} \Vert_{\max}^2$.
Take the median of $6 \log 1/\delta$ evaluations of the mean of
$9\Vert \bm{u}\Vert_1^2 \Vert  \bm{v} \Vert_{\max}^2/(2\epsilon^2)$ samples of $Z$. Then, by using the Chebyshev and Chernoff inequalities, we obtain an $\epsilon$ additive error estimation of $\bm{u} \cdot \bm{v}$
with probability at least $1-\delta$ in $\Ord{\frac{\Vert \bm{u}\Vert_1^2 \Vert  \bm{v} \Vert_{\max}^2}{\epsilon^2}
\log \frac{1}{\delta}}$ queries. 
\end{proof}
\begin{lemma}[Inner product with $\ell_2$-sampling]
\label{lemmaSamplingl2}
Let $\epsilon,\delta \in (0,1)$.
Given query access to $\bm{v} \in\mathbb R^n$ and $\bm{SQ2}(\bm u, n)$ access to  $\bm{u} \in\mathbb R^n$,
we can determine $ \bm{u}\cdot \bm{v}$ to additive error $\epsilon$ with success probability at least $1-\delta$
with $\Ord{\frac{\Vert \bm{u} \Vert_{2}^2 \Vert \bm{v} \Vert_{2}^2}{\epsilon^2} \log \frac{1}{\delta}}$ queries and samples,
and $\tOrd{\frac{\Vert \bm{u} \Vert_{2}^2 \Vert \bm{v} \Vert_{2}^2}{\epsilon^2} \log \frac{1}{\delta}}$ time complexity.
\end{lemma}
\begin{proof}
Define a random variable $Z$ with outcome $\Vert \bm{u} \Vert_2^2 v_j / u_j$ with probability $u_j^2 / \Vert \bm{u} \Vert_2^2$.
Note that $\mathbb E[Z] = \sum_j \Vert \bm{u} \Vert_2^2 v_j u_j^2 / (u_j \Vert \bm{u} \Vert_2^2) = \bm{u} \cdot \bm{v}$.
Also, $\mathbb V[Z] \leq \mathbb E[Z^2] = \sum_j \Vert \bm{u} \Vert_2^2 v_j^2 = \Vert \bm{u} \Vert_2^2 \Vert \bm{v} \Vert_{2}^2$.
Take the median of $6 \log 1/\delta$ evaluations of the mean of $9\Vert \bm{u} \Vert_2^2 \Vert \bm{v}\Vert_{2}^2/(2\epsilon^2)$ samples of $Z$.
Then, by using the Chebyshev and Chernoff inequalities,  we obtain an $\epsilon$ additive error estimation of $\bm{u} \cdot \bm{v}$
with probability at least $1-\delta$ in $\Ord{\frac{\Vert \bm{u}\Vert_2^2 \Vert  \bm{v} \Vert_2^2}{\epsilon^2}
\log \frac{1}{\delta}}$ queries.
\end{proof}

By the above Fact~\ref{factSamplingl2} and Lemma~\ref{lemmaSamplingl2}, given vector $\bm{u}, \bm v \in \mathbb R^n$,
the sampling data structure 
for $\bm{u}$ can be constructed in $\tOrd{n}$ time and
an estimation of $\bm{u} \cdot \bm{v}$ with $\epsilon $ additive error can be obtained 
with probability at least $1 - \delta$  at a run time cost of
$\tOrd {\frac{\Vert \bm{u} \Vert_2 \Vert \bm{v} \Vert_2}{\epsilon^2} {\log \frac{1}{\delta}}}$.

\section{Quantum subroutines}
\label{appQuantum}

First, we recall the quantum access used for vectors.
\begin{defn} [Quantum query access]
Let $c$ and $n$ be two positive integers
and $\bm u$ be
a vector of bit strings $\mathbf{u} \in (\{0,1\}^{c})^n$. 
Define element-wise quantum access to $\bm u$ for $j\in[n]$ by the operation 
\be
\ket j \ket{0^{c}} \to \ket j \ket{ u_j}, 
\ee
on $\Ord{c+ \log n}$ qubits.
We denote this access by $\bm{QA}(\bm u, n,c)$.
\end{defn}
For the following part of this Appendix, recall Definition \ref{defEncoding} regarding the fixed-point encoding of real numbers. 
In addition, we
define the
quantum sample access to a normalized semi-positive vector $\bm v/\Vert \bm v\Vert_1$ which is a fixed-point approximation of a real semi-positive vector.  Each component of the vector $\bm v$ is represented with $c_1$ bits before the decimal point and with $c_2$ bits after the decimal point. 
\begin{defn} [Quantum sample access]
Let $c_1$, $c_2$, and $n$ be positive integers
and $\bm v'\in (\{0,1\}^{c_1})^n$ and $\bm v''\in (\{0,1\}^{c_2})^n$ be
vectors of bit strings. 
Define quantum sample access to a vector $\bm v$ via the operation 
\be
\ket{\bar 0} \to \frac{1}{\sqrt{\Vert \mathcal Q(\bm v', \bm v'')\Vert_1 }} \sum_{j=1}^n  \sqrt{\mathcal Q(v'_j,v''_j)} \ket j,
\ee 
on $\Ord{\log n}$ qubits.
We denote this access by $\bm{QS}(\bm v, n,c_1,c_2)$. For the sample access to a vector $\bm v$ which approximates a vector with components in $[0,1]$, we use the shorthand notation $\bm{QS}(\bm v, n, c_2) := \bm{QS}(\bm v, n,1,c_2)$.
\end{defn}

As stated in \cite{QHedge2020}, 
we have the
following lemma for estimating the $\ell_1$-norm of a vector with entries in $[0,1]$ and preparing states encoding
the square root of the vector elements.
\begin{lemma}[Quantum state preparation and norm estimation] \label{lemmaU}
Let $c$ and $n$ be two positive integers and $\mathbf{u} \in (\{0,1\}^{c+1})^n$. Assume quantum access to $\mathbf{u}$ via $\bm{QA}(\bm u, n,c+1)$.
Let  $\max_j \mathcal Q(u_j) = 1$. 
Then:
\begin{enumerate}
\item There exists a quantum circuit 
that prepares the state
$\frac{1}{\sqrt{n}}  \sum_{j=1}^n \ket j  \left( \sqrt{ \mathcal Q(u_j) } \ket{0} + \sqrt{1- \mathcal Q(u_j)} \ket{1} \right)$ with two queries to $\bm{QA}(\bm u, n,c)$ and $\Ord{\log n + c}$ gates.
\item \label{lemmaU:l1-norm} Let $\epsilon, \delta \in (0,1)$. 
There exists a quantum algorithm that  provides an estimate $\Gamma_{\mathbf{u}}$ of the $\ell_1$-norm
$\Vert \mathcal Q(\mathbf{u}) \Vert_1$  such that
$\left \vert \Vert \mathcal Q(\mathbf{u}) \Vert_1 - \Gamma_{\mathbf{u}}\right \vert \leq \epsilon \Vert \mathcal Q(\mathbf{u})\Vert_1$,
with probability at least $1-\delta$. The quantum circuit of this algorithm
makes $\Ord{\frac{1}{\epsilon}\sqrt{\frac{ n}{\Vert \mathcal Q(\mathbf{u})\Vert_1}} \log(1/\delta)}$  queries to $\bm{QA}(\bm u, n,c+1)$ and has 
$\tOrd{\frac{c}{\epsilon}\sqrt{\frac{ n}{\Vert \mathcal Q(\mathbf{u})\Vert_1}}\log\left (1/\delta\right)}$ gates.
\end{enumerate}
\end{lemma}

By using Lemma~\ref{lemmaU}, we estimate the inner product of two vectors $\bm u$ and $\bm v$ with additive errors as follows.  The vectors can be considered as fixed point approximations to real vectors with elements restricted to $[-1,1]$.

\begin{lemma} [Quantum inner product estimation with additive accuracy] \label{lemmaInnerProduct}
Let $\epsilon,\delta \in(0,1)$.
Let $c$ and $n$ be two positive integers.
Let 
two non-zero vectors of bit strings be $\mathbf{u},\bm v \in (\{0,1\}^{c+2})^n$, which leaves one bit for the sign of each component, one bit for the number before the decimal point, and $c$ bits for the number after the decimal point.  
Let there be given quantum access to $\mathbf{u}$ and $\mathbf{v}$ as $\bm{QA}(\bm u, n,c+2)$ and $\bm{QA}(\bm u, n,c+2)$, respectively. 
Let the norms $\Vert \mathcal Q(\mathbf{u}) \Vert_2$ and $\Vert \mathcal Q(\mathbf{v}) \Vert_2$ be known.
Then,
there exists a quantum algorithm which provides an estimate $I$ for the inner product 
such that
$\vert I - \mathcal Q(\mathbf{u}) \cdot \mathcal Q(\mathbf{v}) \vert \leq \epsilon$ with success probability $1-\delta$.
This estimate is obtained with $\Ord{ \left (\frac{ \Vert \mathcal Q(\mathbf{u}) \Vert_2 \Vert \mathcal Q(\mathbf{v}) \Vert_2 }{\epsilon} + 1 \right) \sqrt n \log \left (\frac{1}{\delta} \right )  }$
queries and $\tOrd{\left (\frac{ \Vert \mathcal Q(\mathbf{u}) \Vert_2 \Vert \mathcal Q(\mathbf{v}) \Vert_2 }{\epsilon} + 1 \right) \sqrt n \log \left (\frac{1}{\delta} \right )  }$ quantum gates.
\end{lemma}
\begin{proof}
Define the 
vectors $\mathbf{u}^{+}$ and $\mathbf{u}^{-}$ as follows
\[ u^+_i := \left \{ \begin{array}{ll}
u_i & {\rm if~} {\rm sign} (u_i) = 1\\
0 &  {\rm otherwise}\\ 
\end{array}\right. \quad \quad \quad \quad
u^-_i = \left \{ \begin{array}{ll}
0 & {\rm if~} {\rm sign} (u_i) = 1\\
-u_i & {\rm otherwise}.\\ 
\end{array}\right.\]
It is easy to see that $\mathcal Q(\mathbf{u}) = \mathcal Q(\mathbf{u}^{+}) - \mathcal Q(\mathbf{u}^{-})$.
Define the vectors $\mathbf{v}^{+}$ and $\mathbf{v}^{-}$ in a similar way.
Then,
\be
\mathcal Q(\mathbf{u}) \cdot \mathcal Q(\mathbf{v}) = \mathcal Q(\mathbf{u}^+) \cdot \mathcal Q(\mathbf{v}^+) + \mathcal Q(\mathbf{u}^{-}) \cdot \mathcal Q(\mathbf{v}^{-})
- \mathcal Q(\mathbf{u}^{+}) \cdot \mathcal Q(\mathbf{v}^{-}) - \mathcal Q(\mathbf{u}^{-}) \cdot \mathcal Q(\mathbf{v}^{+}).
\ee
Define two more vectors of bit strings $\mathbf{z}^+$ and $\mathbf{z}^-$ from
$\mathcal Q(z^+_i) = \mathcal Q(u^+_i)\mathcal Q( v_i^+) + \mathcal Q(u^-_i) \mathcal Q(v^-_i)$ and $\mathcal Q(z^-_i) = \mathcal Q(u^+_i) \mathcal Q(v_i^-) + \mathcal Q(u^-_i v^+_i)$.
Then
\be
\mathcal Q(\mathbf{u}) \cdot \mathcal Q(\mathbf{v}) = \Vert \mathcal Q(\mathbf{z}^+) \Vert_1 - \Vert \mathcal Q(\mathbf{z}^-) \Vert_1.
\ee
In the following, we use the standard $\pm$ notation to denote that a statement holds for both the $+$ and the $-$ case. 
Determine the index of $z^\pm_{\max} := \Vert \mathcal Q(\mathbf{z}^\pm) \Vert_{\max}$ 
with the quantum maximum finding algorithm
with success probability $1 - \delta / 4$,
with $\Ord{ \sqrt{n}  \log \left (\frac{1}{\delta} \right )  }$ queries and
$\tOrd{ \sqrt{n} \log \left (\frac{1}{\delta} \right )  }$ quantum gates \cite{DH96}.
In case that $z^\pm_{\max} = 0$, 
we infer that $ \mathbf{z}^{\pm} = \bm{0}$, 
and if both are true we return the estimate $0$.
Otherwise, for non-zero vector $ \mathbf{z}^{\pm}$, we apply Statement~\ref{lemmaU:l1-norm} of Lemma~\ref{lemmaU} on the vectors of bit strings corresponding to  $\mathcal Q(\mathbf{z}^{\pm}) / z^{\pm}_{\max}$, respectively.
These vectors of bit strings can be computed efficiently from the query access and the result of the maximum finding. 
We obtain estimates $\Gamma^+$ and $\Gamma^-$ such
that
\be
\left \vert \left \Vert \frac{\mathcal Q(\mathbf{z}^\pm)}{ z^\pm_{\max}} \right \Vert_1 - \Gamma^\pm \right \vert \le \epsilon' \left \Vert \frac{\mathcal Q(\mathbf{z}^\pm) }{ z^\pm_{\max} } \right\Vert_1,
\ee
with success probability at least $1 - \delta/4$ for each of them,
with $\Ord{\frac{1}{\epsilon'} \sqrt{\frac{ n z_{\max}^{\pm}}{\Vert \mathcal Q(\mathbf{z}^{\pm})\Vert_1}} \log \left (\frac{1}{\delta} \right )  }$ queries and
$\tOrd{\frac{1}{\epsilon'} \sqrt{\frac{ n z_{\max}^{\pm}}{\Vert \mathcal Q(\mathbf{z}^{\pm})\Vert_1}} \log \left (\frac{1}{\delta} \right )  }$ quantum gates.
Note that  
\be
\Vert \mathcal Q(\mathbf{z}^+) \Vert_1 &=& \mathcal Q(\mathbf{u}^+) \cdot \mathcal Q(\mathbf{v}^+) + \mathcal Q(\mathbf{u}^-) \cdot \mathcal Q(\mathbf{v}^-) \\&\leq& \Vert \mathcal Q(\mathbf{u}^+)\Vert_2 \Vert \mathcal Q(\mathbf{v}^+) \Vert_2 +  \Vert \mathcal Q(\mathbf{u}^-) \Vert_2 \Vert \mathcal Q(\mathbf{v}^-)\Vert_2 \leq 2 \Vert \mathcal Q(\mathbf{u}) \Vert_2 \Vert \mathcal Q(\mathbf{v})\Vert_2,
\ee
where the first inequality follows from Cauchy-Schwarz.
Similarly, we have that $\Vert \mathcal Q(\mathbf{z}^-) \Vert_1 \leq 2 \Vert \mathcal Q(\mathbf{u}) \Vert_2 \Vert \mathcal Q(\mathbf{v})\Vert_2$.

Hence, we obtain an estimate $I = z^+_{\max} \Gamma^+ - z^-_{\max} \Gamma^-$ such that
\[ \begin{array}{lll}
\vert \mathcal Q( \mathbf{u}) \cdot \mathcal Q( \mathbf{v}) - I \vert & = & \vert \Vert \mathcal Q(\mathbf{z}^+) \Vert_1 -
\Vert \mathcal Q(\mathbf{z}^-) \Vert_1 - (z^+_{\max} \Gamma^+ - z^-_{\max}\Gamma^-) \vert \\
& \le & \vert \Vert \mathcal Q(\mathbf{z}^+) \Vert_1 - z^+_{\max} \Gamma^+ \vert +
\vert \Vert \mathcal Q(\mathbf{z}^-) \Vert_1 - z^-_{\max} \Gamma^- \vert \\
& \le & 4\epsilon'\Vert \mathcal Q(\mathbf{u}) \Vert_2 \Vert \mathcal Q (\mathbf{v}) \Vert_2.\\
\end{array}\]
Since $\Vert \mathcal Q(\mathbf{u}) \Vert_2$ and $\Vert \mathcal Q(\mathbf{v}) \Vert_2$ are given, choosing $\epsilon' = \epsilon / (4\Vert \mathcal Q(\mathbf{u}) \Vert_2 \Vert \mathcal Q(\mathbf{v}) \Vert_2)$ 
leads to the result.
The run time of the $\pm$ estimation is then, using $\epsilon' = \epsilon / (4\Vert \mathcal Q(\mathbf{u}) \Vert_2 \Vert \mathcal Q(\mathbf{v}) \Vert_2)$,
\[ \Ord{\frac{\Vert\mathcal Q(\bm u) \Vert_2 \Vert \mathcal Q(\bm v) \Vert_2 \sqrt{n} }{\epsilon} \sqrt{\frac{ z_{\max}^{\pm}}{\Vert \mathcal Q(\mathbf{z}^{\pm})\Vert_1}} \log \left (\frac{1}{\delta} \right )  }.\]
In the absence of more knowledge about the vectors, we take the bound
$z^{\pm}_{\max} / \Vert \mathcal Q(\bm z^{\pm}) \Vert_1 \le 1$. Then the run time is
$\Ord{\frac{\Vert \mathcal Q(\bm u) \Vert_2 \Vert \mathcal Q(\bm v) \Vert_2 \sqrt{n} }{\epsilon} \log \left (\frac{1}{\delta} \right )  }$.
Combining these resource bounds with the resource bounds for maximum finding leads to the stated result. 
\end{proof}
With the following Lemma, we can remove the explicit dimension dependence of the inner product estimation.
For this lemma we suppose 
that one vector is given via quantum query access as before and that 
the other vector is given via access to a quantum subroutine that prepares an amplitude encoding of the vector. In our work, the quantum sampling access  is provided via QRAM in Definition \ref{defn:QRAM}.
The vectors in this lemma are considered to be fixed-point approximations to real vectors with elements restricted to $[0,1]$.
\begin{lemma}[Inner product estimation with quantum sampling and query access] \label{lemmaInnerUnitary}
Let $c$ and $n$ be two positive integers.
Let $\mathbf{u} \in (\{0,1\}^{c+1})^n$ be
a non-zero vector of bit strings, and let $\mathbf{v} \in (\{0,1\}^{c+1})^n$  be another vector of bit strings.
Assume quantum query access to $\bm u$ via $\bm{QA}(\bm u,n,c +1)$, 
and quantum sample access $\bm v$ via $\bm{QS}(\bm v,n,c+1)$.
Then:
\begin{enumerate}[(i)]
\item If $\max_j \mathcal Q(u_j) = 1$, 
then there exists a quantum circuit that prepares the state
$$ \frac{1}{\sqrt{\Vert \mathcal Q(\bm v)\Vert_1}} \sum_{j=1}^n \sqrt{\mathcal Q(v_j)} \ket j  \left( \sqrt{ \mathcal Q(u_j) } \ket{0} + \sqrt{1- \mathcal Q(u_j)} \ket{1} \right)$$ with three queries and $\Ord{c + \log n}$ additional gates.

\item  Let $\epsilon,\delta \in (0,1)$. If $\max_j \mathcal Q(u_j) = 1$, then there exists a quantum algorithm that  provides an estimate $\Gamma$ of
$\frac{\mathcal Q(\bm  v) \cdot \mathcal Q(\bm u)}{\Vert \mathcal Q(\bm v)\Vert_1} $  such that
$\left \vert \frac{\mathcal Q(\bm v) \cdot \mathcal Q(\bm u)}{\Vert \mathcal Q(\bm v)\Vert_1} - \Gamma\right \vert \leq \epsilon$,
with probability at least $1-\delta$. The algorithm requires $\Ord{\frac{1}{\epsilon } \log \frac{1}{\delta}}$  queries and $\tOrd{\frac{1}{\epsilon }   \log\frac{1}{\delta}}$ gates.
\item 
Let $\epsilon,\delta \in (0,1)$. Let the norm $\Vert \mathcal Q(\bm v)\Vert_1$ and $j_{\max} := \arg \max_j \mathcal Q(u_j)$ be known. 
There is a quantum algorithm, similar to (ii), which provides
an estimate $\Gamma'$ of
$ \mathcal Q(\bm v) \cdot \mathcal Q(\bm u)$  such that
$\left \vert \mathcal Q(\bm v) \cdot \mathcal Q(\bm u) - \Gamma'\right \vert \leq \epsilon $,
with probability at least $1-\delta$. 
The algorithm requires $\Ord{\frac{\Vert \mathcal Q(\bm v)\Vert_1 \mathcal Q(u_{j_{\max}}) }{\epsilon } \log \frac{1}{\delta}}$  queries and $\tOrd{\frac{\Vert \mathcal Q(\bm v)\Vert_1 \mathcal Q(u_{j_{\max}})}{\epsilon }   \log \frac{1}{\delta}}$ gates.
\end{enumerate}
\end{lemma}
\begin{proof} 
For (i), with quantum sample access and the quantum query access, perform
\be
\ket {\bar 0} \ket {\bar 0} \ket {0}\to \frac{1}{\sqrt{\Vert  \mathcal Q(\bm v)\Vert_1}}  \sum_{j=1}^n \sqrt{\mathcal Q(v_j)} \ket j  \ket{ \bar 0} \ket { 0} \to \frac{1}{\sqrt{\Vert \mathcal Q(\bm v)\Vert_1}}  \sum_{j=1}^n \sqrt{\mathcal Q(v_j)} \ket j  \ket{u_j} \ket { 0} \\
\to  \frac{1}{\sqrt{\Vert \mathcal Q(\bm v)\Vert_1}}   \sum_{j=1}^N \sqrt{\mathcal Q(v_j)}\ket j  \ket{u_j}  \left( \sqrt{\mathcal Q(u_j)} \ket{0} + \sqrt{1-\mathcal Q(u_j) } \ket{1} \right). \nonumber
\ee
The first step consists of an oracle query to the vector $\bm v$ on the first register.
The second step consists of an oracle query to the vector $\bm u$ which puts the vector component in the second register depending on the index in the first register.
The last step consists of a controlled rotation. The rotation is well-defined as $ \mathcal Q(u_j) \leq \max_j \mathcal Q(u_j) = 1$ 
and can be implemented with 
$\Ord{c}$ gates. Then we uncompute the data register $\ket{u_j } $ with another oracle query.

For (ii),
define a unitary $\mathcal U = U_1 \left(\mathbb I - 2 \ket{\bar 0}\bra{\bar 0}\right) U_1^\dagger$, where $U_1$ is the unitary obtained in (i). Define another unitary by $\mathcal V = \mathbb I-2 \mathbb I \otimes \ket{0} \bra{0}$.
Using $K$ applications of $\mathcal U$ and $\mathcal V$, Amplitude Estimation \cite{Brassard2002} allows to provide an estimation $\tilde a$ of the quantity
$a = \frac{ \mathcal Q(\bm v) \cdot \mathcal Q(\bm u) }{\Vert \mathcal Q(\bm v)\Vert_1 }$
to accuracy
\be
\vert \tilde a - a \vert \leq 2 \pi \frac{\sqrt{a(1-a)} }{K} + \frac{\pi^2}{K^2}.
\ee
Note that $0 \leq a = \frac{ \mathcal Q(\bm v) \cdot \mathcal Q(\bm u) }{\Vert \mathcal Q(\bm v)\Vert_1 } \leq \max_j \mathcal Q(u_j) = 1$.
Set $K> \frac{3 \pi }{\epsilon}$. Then we obtain
\be
\vert \tilde a - a \vert
&\leq& \frac{\pi}{K}\left( 2\sqrt{a}+ \frac{\pi}{K} \right)
< \frac{\epsilon}{3 }  \left( 2\sqrt{a}+  \frac{\epsilon}{3}\right)\leq  \frac{\epsilon}{3 }  3 \leq \epsilon.
\ee
Performing a single run of amplitude estimation with $K$ steps requires
$\Ord{K} =\Ord{\frac{1}{\epsilon }}$
queries to the oracles and
$\Ord{\frac{1}{\epsilon }}$
gates and succeeds with probability $8/\pi^2$. The success probability can be boosted to $1-\delta$ with $\Ord{\log(1/\delta)}$ repetitions of  amplitude estimation. 

For (iii), from the index $j_{\max} = \arg \max_j \mathcal Q(u_j)$ we can obtain the bit string $u_{j_{\max}}$ and its corresponding value $\mathcal Q(u_{j_{\max}})$. This allows us to prepare the quantum circuit for the transformation
\be
\ket {u_j} \ket{0} \to \ket j \ket{ u_j} 
\left(\sqrt{\frac{\mathcal Q(u_j)}{\mathcal Q(u_{j_{\max}})}} \ket{0} + \sqrt{1-\frac{\mathcal Q(u_j)}{\mathcal Q(u_{j_{\max}})} } \ket{1} \right ),
\ee
from the original query access to $\bm u$ and basic arithmetic quantum circuits for the division.
Then we run the same steps as in (ii) with vector $\mathcal Q(\bm u)/\mathcal Q(u_{j_{\max}})$, quantum sample access to vector $\bm v$, and error parameter $\epsilon / (\Vert \mathcal Q(\bm v) \Vert_1 Q(u_{j_{\max}}))$.
We obtain an estimate $\Gamma$ from (ii)
such that
\be
\left \vert \Gamma - \frac{ \mathcal Q(\bm v)}{\Vert \mathcal Q(\bm v)\Vert_1 } \cdot \frac{\mathcal Q(\bm u)}{\mathcal Q(u_{j_{\max}})} \right \vert \le \frac{\epsilon} { \Vert \mathcal Q(\bm v) \Vert_1 \mathcal Q(u_{j_{\max}})}. \label{eq:xf1}
\ee
Then by multiplying both sides of (\ref{eq:xf1}) with $\Vert \mathcal Q(\bm v) \Vert_1 \mathcal Q(u_{j_{\max}})$, we obtain the required estimate
$\Gamma' = \Gamma \Vert \mathcal Q(\bm v) \Vert_1 \mathcal Q(u_{j_{\max}})$.
\end{proof}

\section{Combined algorithm}
\label{appComb}

In this section, we combine the algorithm of the kernel matrix pre-computation and gradient estimation for the learning process. Note that the kernel matrix estimation step introduces an additional error during the process. The classical algorithm is as follows.

\begin{algorithm}
\caption{\textsc{Alphatron\_combined}}\label{alg.alphatron.all}
\SetKw{KwReturn}{Return}
\SetKw{KwInput}{Input}
\SetKw{KwOutput}{Output}
\KwInput{Training data $(\mathbf{x}_i, y_i)^m_{i=1}$ and testing data $(\mathbf{a}_i, b_i)^N_{i=1}$, 
error tolerance parameter
$\epsilon_I$, failure probability $\delta$,
function $u: \mathbb{R} \to [0, 1]$, number of iterations $T$, degree of the multinomial kernel $d$, learning rate $\lambda$}
\DontPrintSemicolon
\;
$\tilde{K}_{ij}, \tilde{K}^{'}_{ij} \gets$ Call \textsc{Alphatron\_with\_Approx\_Pre} (Algorithm~\ref{algAlphatronApprox}) with all inputs as above.\;
Prepare query access to $\tilde{K}_{ij}, \tilde{K}^{'}_{ij}$ as Data Input \ref{inputQueryKernel}.\;
$\alpha^{t_{out}} \gets$ Call  \textsc{Alphatron\_With\_Kernel\_And\_Sampling}(Algorithm~\ref{algAlphatronSamp}) with all inputs as above and $\widetilde{K}_{j}$ and $\widetilde{K}'_{j}$ via Data Input \ref{inputQueryKernel} as prepared.\;
\KwOutput{$\alpha^{t_{out}}$}
\end{algorithm}
The corresponding guarantee and run time is given by the following result. 

\begin{corollary}[Alphatron combined] \label{thm:alphaApproxAll}
Given training data $(\mathbf{x}_i, y_i)^m_{i=1}$ and testing data $(\mathbf{a}_i, b_i)^N_{i=1}$.
Let $K_{\max}$ be maximum of all entries in $K$ and $K'$, which are the corresponding kernel matrices.
Let $\epsilon_I,\delta \in (0,1)$.
If the Definitions \ref{assume:alpha} and \ref{assume:param} hold and $A_2 \leq 1$ and with $\epsilon_I=\sqrt{\epsilon}$, the Algorithm \ref{alg.alphatron.all} outputs $\alpha^{t_{\rm out}}$ which describes the hypothesis $h^{t_{\rm out}}(\bm{x}) := u\left (\sum_{i=1}^{m_1} \alpha_i^{t_{\rm out}} \mathcal{K}_d(\mathbf{x}, \mathbf{x}_i)\right)$ such that with probability $1 - 4\delta$,
\[
\varepsilon(h^{t_{\rm out}}) \in \Ord{A_2}.
\]
The run time for obtaining this output is in 
$$\tOrd{mn+ \frac{T^2m^2d^2}{\epsilon}\log\left(\frac{1}{\delta}\right)+Tm + T^3 m \frac{K_{\max}^2 }{L^2 \epsilon} \log \left( \frac{1}{\delta}\right)}.$$
\end{corollary}

\begin{proof}
The proof is the same as Theorem \ref{thm:alphaApproxLoop}, except for the inequalities as follows.
Here, we use
\textsc{Alphatron\_with\_Approx\_Pre} to estimate the kernel matrices with precision $\epsilon_K$, which are denoted by 
$\tilde{K}_{ij}$ and $\tilde{K}'_{ij}$ respectively.
Therefore, with \textsc{Alphatron\_With\_Kernel\_And\_Sampling}, we estimate the inner product $\alpha^t \cdot \tilde{K}_j$ instead of $\alpha^t \cdot K_j$, to additive accuracy $\epsilon'_I$.
Let the estimated value be $r^t_j$.
To promise
\be
\left|r^t_j-\sum_{i=1}^{m_1}\alpha^t_i K_{ji}\right|&=&\left|r^t_j-\sum_{i=1}^{m_1}\alpha^t_i \tilde{K}_{ji}+\sum_{i=1}^{m_1}\alpha^t_i \tilde{K}_{ji}-\sum_{i=1}^{m_1}\alpha^t_i K_{ji}\right|\\
&\leq& \left|r^t_j-\sum_{i=1}^{m_1}\alpha^t_i \tilde{K}_{ji}\right|+\left|\sum_{i=1}^{m_1}\alpha^t_i \tilde{K}_{ji}-\sum_{i=1}^{m_1}\alpha^t_i K_{ji}\right|\\
&\leq& \epsilon_I'+ 2\lambda T\epsilon_K\leq \epsilon_I,
\ee
it suffices to take $\epsilon'_I=\epsilon_I/2$ and $\epsilon_K=\epsilon_I/(4\lambda T)$.
Recall that by Lemma \ref{lem:alphabound}, $0<\alpha_{\max}\leq \lambda T/m_1$.
For the run time,  storing $\tilde{K}_{ij}$ and $\tilde{K}'_{ij}$ takes time $\Ord{m^2}$, which is included in the quoted run time.  
Set $\epsilon_I=\sqrt{\epsilon}$, we obtain the run time from Theorem \ref{thm:alphaApproxLoop}.
\end{proof}
We continue with the corresponding quantum algorithm. 
For the result, we slightly change the proof for Theorem \ref{thm:alphaQuantumLoop}.

\begin{algorithm}
\caption{\textsc{Quantum\_Alphatron\_combined}}\label{alg.quantumalphatron.all}
\SetKw{KwReturn}{Return}
\SetKw{KwInput}{Input}
\SetKw{KwOutput}{Output}
\KwInput{Quantum access to training data $(\mathbf{x}_i, y_i)^m_{i=1}$ and testing data $(\mathbf{a}_i, b_i)^N_{i=1}$ according to Data Input \ref{inputQuantumAccess1}, 
error tolerance parameter
$\epsilon_I$, failure probability $\delta$, 
function $u: \mathbb{R} \to [0, 1]$, number of iterations $T$, degree of the multinomial kernel $d$, learning rate $\lambda$}
\DontPrintSemicolon
\;
$\tilde{K}_{ij}, \tilde{K}^{'}_{ij} \gets$ Call \textsc{Alphatron\_with\_Q\_Pre}(Algorithm~\ref{algAlphatronQuantumPre}) with all inputs as above.\;
Store $\tilde{K}_{ij}, \tilde{K}^{'}_{ij}$ via quantum random access memory (QRAM) and prepare the quantum query access to $\tilde{K}_j, \tilde{K}^{'}_j$ as Data Input \ref{input:quantumkernel}.\;
$\alpha^{t_{out}} \gets$ Call  \textsc{Quantum\_Alphatron}(Algorithm~\ref{algQAlphatron2}) with all inputs as above and $\widetilde{K}_{j}$ and $\widetilde{K}'_{j}$ via Data Input \ref{input:quantumkernel} as prepared.\;
\KwOutput{$\alpha^{t_{out}}$}
\end{algorithm}

\begin{corollary}[Quantum Alphatron combined] \label{thm:alphaQuantumAll}
We assume quantum query access to the vectors $\mathbf{x}_i$ and $\mathbf{a}_i$ via Data Input~\ref{input:quantumkernel}.
Let $K_{\max}$ be maximum of all 
entries in $K$ and $K'$, which are the corresponding kernel matrices.
Let $\delta \in (0,1)$.
Given Definitions \ref{assume:alpha} and \ref{assume:param} with $A_2 \le 1$ and with $\epsilon_I=\sqrt{\epsilon}$, Algorithm \ref{alg.quantumalphatron.all}
outputs $\alpha^{t_{out}}$ such that the hypothesis $h^{t_{\rm out}} = u\left(\sum_j\alpha_j^{t_{\rm out}} \psi(\bm x_j) \cdot \psi(\bm x)\right)$ satisfies with probability $1-4\delta$,
\[
\varepsilon(h^{t_{\rm out}}) \in \Ord{A_2}.
\]
The run time of this algorithm is 
\[ \tOrd{\frac{Tm^2d\sqrt{n}}{\sqrt{\epsilon}}\log\left(\frac{1}{\delta}\right)+  m^{1.5}\log \left( \frac{1}{\delta} \right) + Tm + T^2 m \frac{ K_{\max} }{L \sqrt \epsilon} \log \left( \frac{1}{\delta}\right) }.\]
\end{corollary}
\begin{proof}
The proof is the same with Theorem \ref{thm:alphaQuantumLoop}, except that instead of Eq.~(\ref{eqqnum}), we consider
\be
\left|\alpha^t\cdot K_i-\tilde{\alpha}^t\cdot \tilde{K}_i\right|&=&\left|\alpha^t\cdot K_i-\tilde{\alpha}^t\cdot K_i+\tilde{\alpha}^t\cdot K_i-\tilde{\alpha}^t\cdot \tilde{K}_i\right|\notag\\
&\leq& \left |\alpha^t\cdot K_i-\tilde{\alpha}^t\cdot K_i\right |+\left|\tilde{\alpha}^t\cdot K_i-\tilde{\alpha}^t\cdot \tilde{K}_i\right|\notag\\
&\leq& \frac{m_1 K_{\max}}{2^{c_2+1}}+\tilde{\alpha}^tm_1\epsilon_K\notag\\
&\leq& \frac{m_1 K_{\max}}{2^{c_2+1}}+2\lambda T \epsilon_K\leq \frac{\epsilon_I}{2},
\ee
where the last inequality can be achieved by setting $c_2 \geq \left \lceil \log\left(\frac{2K_{\max} m_1}{\epsilon_I} \right) +1\right \rceil$, $\epsilon_K\leq \epsilon_I/(8\lambda T)$. 
Recall that by Lemma \ref{lem:alphabound}, $0<\alpha_{\max}\leq \lambda T/m_1$.
The rest follows as in Theorem \ref{thm:alphaQuantumLoop}.
Storage in QRAM costs $\Ord{m^2}$.
Set $\epsilon_I=\sqrt{\epsilon}$, we obtain the run time from Theorem \ref{thm:alphaQuantumLoop}.
\end{proof}

\section{Rademacher complexity}
\label{appRade}
The following standard generalization bound based on Rademacher complexity is employed in 
our analysis. For a background on Rademacher
complexity, we refer the reader to \cite{Bartlett2002}.
\begin{theorem}[Generalization bound~\cite{Bartlett2002}] \label{generalizationbound}
Let $\mathcal{D}$ be a distribution over $\mathcal{X} \times \mathcal{Y}$ and let $\mathcal{L} : \mathcal{Y}^\prime
\times \mathcal{Y} \to [-b, b]$ (where $\mathcal{Y} \subseteq \mathcal{Y}^\prime \subseteq \mathbb{R}$) be a
$b$-bounded loss function that is $L$-Lipschitz in its first argument.  Let
$\mathcal{F} \subseteq (\mathcal{Y}^\prime)^\mathcal{X}$ and for any $f \in \mathcal{F}$, let $\mathcal{J}(f, \mathcal{D}) = \E_{(\textbf{x}, y)
\sim \mathcal{D}}[\mathcal{L}(f(\textbf{x}), y)]$ and $\hat{\mathcal{J}}(f, S) = \frac{1}{m} \sum_{i = 1}^m
\mathcal{L}(f(\textbf{x}_\textbf{i}), y_i)$, where $S = ((\textbf{x}_\textbf{1}, y_1), \ldots,  (\textbf{x}_\textbf{m}, y_m))  \sim
\mathcal{D}^m$. Then for any $\delta > 0$, with probability at least $1 - \delta$
(over the random sample draw for $S$), simultaneously for all $f \in
\mathcal{F}$, the following is true:
\[
|\mathcal{J}(f, \mathcal{D}) - \hat{\mathcal{J}}(f, S)| \leq 4 \cdot L \cdot \mathcal{R}_\textbf{m}(\mathcal{F})
+ 2\cdot b \cdot \sqrt{\frac{\log (1/\delta)}{2m}}
\]
where $\mathcal{R}_\textbf{m}(\mathcal{F})$ is the Rademacher complexity of the function class $\mathcal{F}$. 
\label{thm:Rademacher}
\end{theorem}

\printbibliography
\end{document}